\documentclass[11pt]{article}

\usepackage[utf8]{inputenc}
\usepackage[T1]{fontenc}
\usepackage{amsmath,amssymb,amsthm}
\usepackage{mathtools}
\usepackage{bm}
\usepackage{graphicx}
\usepackage{hyperref}
\usepackage{cleveref}
\usepackage{natbib}
\usepackage[margin=1in]{geometry}
\usepackage{enumitem}
\usepackage{booktabs}
\usepackage{microtype}
\usepackage{xcolor}
\usepackage{tocloft}
\hypersetup{colorlinks=true, linkcolor=blue, citecolor=blue, urlcolor=blue}

\newtheorem{theorem}{Theorem}[section]
\newtheorem{proposition}[theorem]{Proposition}
\newtheorem{lemma}[theorem]{Lemma}
\newtheorem{corollary}[theorem]{Corollary}
\newtheorem{definition}[theorem]{Definition}
\newtheorem{remark}[theorem]{Remark}

\newcommand{\R}{\mathbb{R}}
\newcommand{\E}{\mathbb{E}}
\newcommand{\KL}{\mathrm{KL}}
\newcommand{\TV}{\mathrm{TV}}

\newcommand{\rank}{\mathrm{rank}}

\newcommand{\Var}{\mathrm{Var}}
\newcommand{\Cov}{\mathrm{Cov}}
\newcommand{\sech}{\mathrm{sech}}
\DeclareMathOperator*{\argmin}{arg\,min}
\DeclareMathOperator*{\argmax}{arg\,max}

\title{\textbf{Score Shocks: The Burgers Equation Structure\\of Diffusion Generative Models}}
\author{Krisanu Sarkar\\Indian Institute of Technology Bombay\\Mumbai, India}
\date{}

\begin{document}
\maketitle
\setcounter{tocdepth}{2}
\tableofcontents

\newpage

\begin{abstract}
We analyze the score field of a diffusion generative model through a Burgers-type evolution law. For VE diffusion, the heat-evolved data density implies that the score obeys viscous Burgers in one dimension and the corresponding irrotational vector Burgers system in $\R^d$, giving a PDE view of \emph{speciation transitions} as the sharpening of inter-mode interfaces. For any binary decomposition of the noised density into two positive heat solutions, the score separates into a smooth background and a universal $\tanh$ interfacial term determined by the component log-ratio; near a regular binary mode boundary this yields a normal criterion for speciation. In symmetric binary Gaussian mixtures, the criterion recovers the critical diffusion time detected by the midpoint derivative of the score and agrees with the spectral criterion of Biroli, Bonnaire, de~Bortoli, and M\'ezard (2024). After subtracting the background drift, the inter-mode layer has a local Burgers $\tanh$ profile, which becomes global in the symmetric Gaussian case with width $\sigma_\tau^2/a$. We also quantify exponential amplification of score errors across this layer, show that Burgers dynamics preserves irrotationality, and use a change of variables to reduce the VP-SDE to the VE case, yielding a closed-form VP speciation time. Gaussian-mixture formulas are verified to machine precision, and the local theorem is checked numerically on a quartic double-well.
\end{abstract}

\section{Introduction}\label{sec:intro}

Diffusion generative models are now a standard paradigm in modern machine learning, with strong results in image synthesis~\citep{dhariwal2021diffusion, rombach2022high}, video generation~\citep{ho2022video}, audio synthesis, and scientific applications ranging from molecular design to weather prediction.
The framework, introduced by \citet{sohl2015deep} and developed into its modern form by \citet{song2019generative}, \citet{ho2020denoising}, and \citet{song2021score}, rests on two complementary processes.
The \emph{forward} process gradually corrupts data with noise according to a stochastic differential equation (SDE), transforming any data distribution into an approximately Gaussian prior.
The \emph{reverse} process inverts this corruption by learning the \emph{score function}---the gradient of the log-density of the noised data, $\nabla_{\bm{x}} \log p_t(\bm{x})$---and using it to drive a reverse-time SDE~\citep{anderson1982reverse} or a deterministic probability flow ODE~\citep{song2021score}.

Despite their empirical triumph, the mathematical structures governing the score function's behavior during the generative process remain only partially understood.
A growing body of work in statistical physics has revealed that the reverse-time dynamics of diffusion models exhibit \emph{phase transitions}: moments at which generative trajectories spontaneously commit to distinct data modes through a mechanism akin to symmetry breaking in equilibrium systems~\citep{raya2023spontaneous, biroli2023generative, biroli2024dynamical, ambrogioni2025thermodynamics}.
\citet{biroli2024dynamical} identified three dynamical regimes---a noise-dominated regime, a \emph{speciation} transition where coarse class structure emerges, and a \emph{collapse} transition where trajectories lock onto individual training points---and characterized these using mean-field methods from spin-glass theory.
Concurrently, \citet{sclocchi2024phase} showed that hierarchical data structure is revealed through successive phase transitions, while \citet{li2024critical} established non-asymptotic ``critical window'' bounds for feature emergence.
On the PDE side, \citet{lai2023fp} derived a Fokker--Planck equation governing the evolution of the score and used it as a training regularizer, and very recently \citet{vuong2025score} demonstrated empirically that trained score networks produce non-conservative vector fields, reinterpreting diffusion models through the lens of Wasserstein gradient flows.

\paragraph{Contribution.}
We study the score field of a diffusion model through its Burgers structure.
In one dimension, the score of any VE diffusion satisfies the viscous Burgers equation exactly; in $\R^d$, it satisfies the corresponding vector Burgers system.
This follows directly from the Cole--Hopf transform applied to the heat equation governing the forward process.
The results fall into four levels: a Burgers correspondence for general diffusions, a local binary-boundary theorem for arbitrary smooth densities, closed-form statements for symmetric binary Gaussian mixtures, and asymptotic or corrected criteria for more general asymmetric settings.
This leads to several concrete consequences:

\begin{enumerate}[label=(\roman*),leftmargin=2em]
\item \textbf{Speciation threshold and Burgers interpretation.}
At any regular binary mode boundary, the normal Hessian decomposes as $\partial_n s_n = \partial_n \bar s_n + \kappa^2/4$, separating a smooth background term from a universal positive interfacial contribution.
For symmetric binary mixtures, this local criterion reduces to the midpoint derivative condition $s_x(0,\tau^\ast)=0$ and agrees with the spectral criterion of \citet{biroli2024dynamical}.

\item \textbf{Interfacial profile at mode boundaries.}
For any binary heat decomposition, the background-subtracted normal score has a local $\tanh$ interfacial profile in boundary-normal coordinates; in the symmetric Gaussian case, after removing the ambient Gaussian drift, the profile is global and its width is explicit.

\item \textbf{Error amplification.}
Score-estimation errors are amplified near the interfacial layer by a factor $\exp(\Lambda)$ with $\Lambda \approx \mathrm{SNR}/2$, giving a PDE-theoretic explanation for the sensitivity of sample quality to low-noise score accuracy~\citep{song2020improved, karras2022elucidating}.

\item \textbf{Curl preservation.}
The vector Burgers dynamics preserves irrotationality, so the non-conservative components observed by \citet{vuong2025score} in trained networks are attributable to approximation error rather than to the underlying dynamics.

\item \textbf{VP-to-VE reduction.}
A coordinate transformation reduces the VP-SDE (Ornstein--Uhlenbeck) score equation to the pure VE Burgers case, yielding closed-form VP speciation times and interfacial widths within a single analytical framework.
\end{enumerate}

All formal statements are proved in the text.
The Gaussian-mixture predictions are verified to machine precision (${\sim}10^{-9}$), and the general local theorem is checked numerically on a non-Gaussian quartic double-well.
The numerical checks are modest in scale and are included mainly to verify the formulas stated above.

\paragraph{Organization.}
\Cref{sec:related,sec:prelim} collect related work and notation. The Burgers correspondence is derived in \Cref{sec:score-burgers}, and the interfacial theory is developed in \Cref{sec:shocks}. The later sections treat error amplification, higher-dimensional extensions, the VP reduction, correction terms, numerical checks, and concluding remarks.

\section{Related Work}\label{sec:related}

Our work sits at the intersection of three lines of research: the mathematical theory of diffusion generative models, the statistical physics of generative processes, and the classical PDE theory of the Burgers equation.
We survey each in turn, emphasizing the gaps that our contribution fills.

\subsection{Score-Based Diffusion Models}

The idea of generating samples by learning the score function and running Langevin dynamics was introduced by \citet{song2019generative}, building on the score matching framework of \citet{hyvarinen2005estimation} and the denoising score matching perspective of \citet{vincent2011connection}.
\citet{ho2020denoising} developed Denoising Diffusion Probabilistic Models (DDPMs), connecting the forward process to a discrete Markov chain and training via a reweighted variational bound.
The continuous-time unification came with \citet{song2021score}, who showed that both the Noise Conditional Score Network (NCSN) framework and DDPM are discretizations of forward and reverse SDEs, with the reverse dynamics depending on the score through the celebrated result of \citet{anderson1982reverse}.
This SDE perspective enabled the derivation of deterministic samplers (the probability flow ODE), exact likelihood computation, and principled noise schedule design~\citep{song2021ddim, kingma2021variational, karras2022elucidating}.

The convergence theory of diffusion models has advanced rapidly.
\citet{debortoli2022convergence} established convergence under the manifold hypothesis.
\citet{chen2023sampling} proved polynomial-time sampling guarantees under minimal assumptions, showing that the total variation distance between the generated and true distributions is controlled by the $L^2$ score estimation error integrated over time.
\citet{benton2024nearly} sharpened these bounds using stochastic localization, achieving nearly $d$-linear convergence.
\citet{lee2023convergence} and \citet{tang2024score} provided accessible surveys of the theoretical landscape.

Our work complements this literature by revealing the \emph{PDE structure} of the score itself.
While the convergence theory treats the score as a generic vector field and bounds the effect of estimation error, our Burgers equation framework shows \emph{where} in space-time the score is most fragile (at the shocks) and \emph{why} (the classical gradient blowup of inviscid Burgers), providing geometric insight that the $L^2$-based bounds do not capture.

\subsection{Phase Transitions and Symmetry Breaking in Diffusion}

A parallel line of investigation, rooted in statistical physics, has uncovered the dynamical phase structure of diffusion models.
\citet{raya2023spontaneous} first identified spontaneous symmetry breaking in the reverse generative process: at a critical noise level, the score field bifurcates and generative trajectories commit to distinct modes.
\citet{biroli2023generative} analyzed this phenomenon in very high dimensions using methods from random matrix theory and the statistical mechanics of disordered systems, showing that speciation occurs at a noise level determined by the spectrum of the data covariance.
The comprehensive framework of \citet{biroli2024dynamical}, published in \emph{Nature Communications}, delineated three dynamical regimes---noise-dominated, speciation, and collapse---and characterized the speciation crossover through a spectral analysis of the empirical covariance, with the collapse transition governed by an ``excess entropy'' quantity reminiscent of the glass transition.

These results were extended in several directions.
\citet{sclocchi2024phase} demonstrated that hierarchical data gives rise to a cascade of speciation transitions, each corresponding to the emergence of progressively finer structure.
\citet{li2024critical} provided non-asymptotic ``critical window'' bounds for the time interval during which features emerge, complementing the asymptotic analysis of \citet{biroli2024dynamical}.
\citet{ambrogioni2025thermodynamics} reformulated the entire framework in the language of equilibrium statistical thermodynamics, defining a free energy landscape whose minima correspond to data modes and whose phase transitions are mean-field in character.
Very recently, \citet{ambrogioni2025information} connected the score divergence to entropy production rates and showed that the variance of pathwise conditional entropy peaks at the speciation time, providing an information-theoretic diagnostic.
On the memorization side, \citet{bonnaire2025memorization} and \citet{achilli2025memorization} studied how approximate score learning prevents the collapse transition in practical models.

Our contribution provides a \emph{PDE-theoretic} counterpart to these statistical-physics results.
The paper is organized around a nested hierarchy of statements. First, the score of any VE diffusion obeys Burgers exactly. Second, near any regular binary mode boundary, the score admits an exact decomposition into a smooth background plus a universal $\tfrac{1}{2}\tanh(\phi/2)\nabla\phi$ layer. Third, in Gaussian mixture models these general structures become explicit formulas for the threshold, profile, amplification exponent, and boundary motion.
This hierarchy is what allows the symmetric Gaussian case to serve both as a solvable model and as a faithful specialization of the more general local theorem.

\subsection{The Score PDE and Non-Conservative Learned Scores}

The PDE governing the time evolution of the score was derived by \citet{lai2023fp}, who termed it the ``score Fokker--Planck equation'' and showed that enforcing it as a regularizer during training improves both log-likelihood and the conservativity (curl-freeness) of the learned score.
Their empirical observation that trained scores have non-negligible curl was strikingly confirmed by \citet{vuong2025score}, who showed that trained diffusion networks violate both integral and differential constraints required of gradient fields, and proposed reinterpreting diffusion models as learning velocity fields of Wasserstein gradient flows~\citep{jordan1998variational, ambrosio2005gradient} rather than true score functions.

Our work places these findings in a unified PDE framework.
We note in particular that the ``score Fokker--Planck equation'' of \citet{lai2023fp} can be written as the viscous Burgers equation after the identification $u = -2s$.
A later section proves (\Cref{thm:curl-preservation}) that the Burgers dynamics preserves irrotationality: the vorticity $\omega_{ij} = \partial_i s_j - \partial_j s_i$ satisfies a linear parabolic equation with zero initial data and therefore remains zero.
This provides a theoretical guarantee that the curl observed by \citet{vuong2025score} and \citet{lai2023fp} cannot come from the exact Burgers dynamics itself; within our framework, it must arise from approximation, discretization, or modeling error.
Furthermore, we connect the non-conservative components to the theory of entropy-violating weak solutions of the Burgers equation~\citep{lax1957hyperbolic}, suggesting a practical diagnostic for score network quality.

\subsection{The Burgers Equation}

The Burgers equation $u_t + u\,u_x = \nu\,u_{xx}$ was introduced by \citet{burgers1948mathematical} as a simplified model of turbulence and has since become one of the canonical nonlinear PDEs in mathematical physics.
The seminal discovery that it can be linearized via the Cole--Hopf transformation---independently by \citet{hopf1950partial} and \citet{cole1951quasi}---reduces it to the heat equation and enables exact solutions for arbitrary initial data.
In the inviscid limit $\nu \to 0$, smooth solutions break down in finite time through the formation of \emph{shocks}---discontinuities across which the Rankine--Hugoniot conditions~\citep{rankine1870thermodynamic, hugoniot1889propagation} determine the jump relations.
The selection of physically relevant (entropy-satisfying) weak solutions is governed by the Lax entropy condition~\citep{lax1957hyperbolic}.
The comprehensive treatment by \citet{whitham1974linear} and the modern PDE perspective of \citet{evans2010partial} provide the mathematical foundations we employ.

The Burgers equation arises naturally in the study of the Kardar--Parisi--Zhang (KPZ) equation for interface growth~\citep{kardar1986dynamic} and appears throughout fluid dynamics, cosmology, and traffic flow modeling.
Its appearance in the context of diffusion generative models, however, helps organize several phenomena that are otherwise studied separately.
The Cole--Hopf transform links the score $s = \partial_x \log p$ to the Burgers velocity $u = -2s$ whenever $p$ satisfies the heat equation---a mathematically elementary observation whose implications for generative modeling are developed in the sections that follow.

\subsection{Stochastic Localization and Optimal Transport}

A related analytical framework uses \emph{stochastic localization}~\citep{elalaoui2022sampling} to study the convergence of diffusion-based sampling algorithms.
\citet{montanari2023sampling} showed that stochastic localization provides an elegant generalization of diffusion models, and \citet{benton2024nearly} leveraged this connection for sharp convergence bounds.
The stochastic interpolant framework of \citet{albergo2023stochastic} and the flow matching perspective of \citet{lipman2023flow, liu2023flow} provide further connections between diffusion, optimal transport, and score-based generation.

Our Burgers equation perspective is complementary to these approaches.
Stochastic localization analyzes the convergence of the \emph{distribution} to the target; the Burgers framework analyzes the \emph{dynamics of the score field itself}, revealing its singularity structure.
The two perspectives meet at the speciation time, which appears as a critical localization time in the stochastic localization framework and as a shock-like threshold in the Burgers framework.

\bigskip
\noindent
Taken together, these ingredients connect the Burgers equation with the score field of diffusion generative models. The tools involved are standard---the Cole--Hopf transform, Rankine--Hugoniot conditions, and Gr\"onwall bounds---but their combination gives a direct link between the PDE and statistical-physics viewpoints used later in the paper.

\section{Preliminaries}\label{sec:prelim}

We fix notation and recall the SDE framework for diffusion generative models, following the unified treatment of \citet{song2021score}.
Throughout, we work in $\R^d$ (with $d=1$ made explicit where the one-dimensional theory is invoked) and use Einstein summation convention only when stated.

\subsection{Forward diffusion processes}\label{sec:prelim:forward}

A diffusion generative model is defined by a forward SDE that progressively corrupts data into noise:
\begin{equation}\label{eq:forward-sde}
  d\bm{X}_t = \bm{f}(\bm{X}_t, t)\,dt + g(t)\,d\bm{W}_t, \qquad \bm{X}_0 \sim p_0,
\end{equation}
where $\bm{f}\colon \R^d \times [0,T] \to \R^d$ is the drift, $g\colon [0,T] \to \R_{>0}$ is the scalar diffusion coefficient, $\bm{W}_t$ is a standard $d$-dimensional Wiener process, and $p_0$ is the data distribution~\citep{song2021score}.
The marginal density $p(\bm{x},t)$ of $\bm{X}_t$ satisfies the Fokker--Planck equation (FPE):
\begin{equation}\label{eq:fpe}
  \frac{\partial p}{\partial t}
  = -\nabla \cdot (\bm{f}\, p)
    + \frac{g(t)^2}{2}\,\Delta p.
\end{equation}
We consider two standard instantiations.

\paragraph{Variance-Exploding (VE) SDE.}
Setting $\bm{f} = \bm{0}$, the forward process is pure diffusion~\citep{song2019generative, song2020improved}:
\begin{equation}\label{eq:ve-sde}
  d\bm{X}_t = g(t)\,d\bm{W}_t, \qquad \bm{X}_0 \sim p_0.
\end{equation}
The FPE reduces to the heat equation with time-dependent diffusivity $\nu(t) = g(t)^2/2$:
\begin{equation}\label{eq:ve-fpe}
  \frac{\partial p}{\partial t} = \nu(t)\,\Delta p.
\end{equation}
The conditional distribution is $\bm{X}_t \mid \bm{X}_0 \sim \mathcal{N}(\bm{X}_0,\, \sigma^2_{\mathrm{VE}}(t)\,\bm{I})$, where $\sigma^2_{\mathrm{VE}}(t) = \int_0^t g(s)^2\,ds$.

\paragraph{Variance-Preserving (VP) SDE.}
Setting $\bm{f}(\bm{x},t) = -\tfrac{1}{2}\beta(t)\bm{x}$ with a positive schedule $\beta(t)$ yields the Ornstein--Uhlenbeck (OU) forward process~\citep{ho2020denoising, song2021score}:
\begin{equation}\label{eq:vp-sde}
  d\bm{X}_t = -\tfrac{1}{2}\beta(t)\,\bm{X}_t\,dt + \sqrt{\beta(t)}\,d\bm{W}_t.
\end{equation}
Define the signal attenuation $\alpha(t) = \exp\!\bigl(-\tfrac{1}{2}\int_0^t \beta(s)\,ds\bigr)$.
Then $\bm{X}_t \mid \bm{X}_0 \sim \mathcal{N}\!\bigl(\alpha(t)\bm{X}_0,\, (1-\alpha(t)^2)\bm{I}\bigr)$, and the FPE is:
\begin{equation}\label{eq:vp-fpe}
  \frac{\partial p}{\partial t}
  = \frac{\beta(t)}{2}\,\nabla \cdot (\bm{x}\, p)
    + \frac{\beta(t)}{2}\,\Delta p.
\end{equation}

\subsection{Diffusion-time reparametrization}\label{sec:prelim:tau}

For the VE-SDE, define the \emph{cumulative diffusion time}:
\begin{equation}\label{eq:tau-def}
  \tau(t) = \frac{1}{2}\int_0^t g(s)^2\,ds = \frac{\sigma^2_{\mathrm{VE}}(t)}{2}.
\end{equation}
Under this change of variable, $d\tau = \nu(t)\,dt$, and \eqref{eq:ve-fpe} becomes the \emph{standard heat equation} with unit diffusion coefficient:
\begin{equation}\label{eq:heat}
  \frac{\partial p}{\partial \tau} = \Delta p.
\end{equation}
We write $p_\tau(\bm{x}) \equiv p(\bm{x},\tau)$.
The solution is the convolution $p_\tau = p_0 * G_\tau$, where $G_\tau(\bm{x}) = (4\pi\tau)^{-d/2}\exp(-|\bm{x}|^2/(4\tau))$ is the heat kernel~\citep{evans2010partial}.
For $\tau > 0$, strict positivity of $G_\tau$ ensures $p_\tau(\bm{x}) > 0$ for all $\bm{x} \in \R^d$, so that $\log p_\tau$ and all its derivatives are well-defined and smooth.

\begin{remark}\label{rem:tau-convention}
Unless otherwise stated, all analysis in \Cref{sec:score-burgers,sec:shocks} is conducted in $\tau$-time with the VE-SDE.
The extension to physical time $t$ is recovered by the substitution $\partial_\tau \mapsto \nu(t)^{-1}\partial_t$, and the VP case is treated in \Cref{sec:vp} via a coordinate transformation that reduces it to the VE setting.
\end{remark}

\subsection{The score function}\label{sec:prelim:score}

\begin{definition}[Score function]\label{def:score}
The \emph{score function} of the noised density $p_\tau$ is the vector field
\begin{equation}\label{eq:score-def}
  \bm{s}(\bm{x},\tau)
  = \nabla_{\bm{x}} \log p_\tau(\bm{x})
  = \frac{\nabla p_\tau(\bm{x})}{p_\tau(\bm{x})}.
\end{equation}
In one dimension ($d=1$), we write $s(x,\tau) = \partial_x \log p_\tau(x)$.
\end{definition}

The score is the central object in score-based generative modeling~\citep{song2019generative, hyvarinen2005estimation}.
It determines the reverse-time SDE~\citep{anderson1982reverse}
\begin{equation}\label{eq:reverse-sde}
  d\bm{X}_t = \bigl[\bm{f}(\bm{X}_t,t) - g(t)^2 \bm{s}(\bm{X}_t,t)\bigr]\,dt + g(t)\,d\bar{\bm{W}}_t,
\end{equation}
where $\bar{\bm{W}}_t$ is a reverse-time Wiener process, and the deterministic \emph{probability flow ODE}~\citep{song2021score}
\begin{equation}\label{eq:prob-flow}
  \frac{d\bm{x}}{dt} = \bm{f}(\bm{x},t) - \frac{g(t)^2}{2}\,\bm{s}(\bm{x},t).
\end{equation}
In practice, a neural network $\bm{s}_\theta(\bm{x},t)$ is trained to approximate $\bm{s}$ via the denoising score matching objective~\citep{vincent2011connection, song2021score}:
\begin{equation}\label{eq:dsm}
  \mathcal{L}(\theta)
  = \E_{t}\E_{\bm{X}_0}\E_{\bm{X}_t | \bm{X}_0}
    \bigl[\lambda(t)\,\|\bm{s}_\theta(\bm{X}_t, t) - \nabla_{\bm{X}_t}\log p(\bm{X}_t \mid \bm{X}_0)\|^2\bigr],
\end{equation}
where $\lambda(t)$ is a positive weighting function.

\subsection{Notation for Gaussian mixtures}\label{sec:prelim:gmm}

Our main analytical results concern data distributions that are finite Gaussian mixtures:
\begin{equation}\label{eq:gmm-def}
  p_0(\bm{x}) = \sum_{k=1}^K w_k\,\mathcal{N}(\bm{x};\,\bm{\mu}_k,\,\sigma_0^2 \bm{I}_d),
\end{equation}
with weights $w_k > 0$ summing to one, means $\bm{\mu}_k \in \R^d$, and common component variance $\sigma_0^2$.
Under the VE forward process at diffusion time $\tau$, the noised density is
\begin{equation}\label{eq:gmm-noised}
  p_\tau(\bm{x}) = \sum_{k=1}^K w_k\,\mathcal{N}(\bm{x};\,\bm{\mu}_k,\,\sigma_\tau^2 \bm{I}_d),
  \qquad \sigma_\tau^2 \coloneqq \sigma_0^2 + 2\tau.
\end{equation}
We define the weighted mean $\bar{\bm{x}} = \sum_k w_k \bm{\mu}_k$, centered means $\bm{\nu}_k = \bm{\mu}_k - \bar{\bm{x}}$, and the \emph{between-class covariance}
\begin{equation}\label{eq:W-def}
  \bm{W} = \sum_{k=1}^K w_k\,\bm{\nu}_k \bm{\nu}_k^\top,
\end{equation}
which is positive semidefinite with $\rank(\bm{W}) \leq \min(K-1, d)$.
Its eigenvalues $\lambda_1 \geq \lambda_2 \geq \cdots \geq \lambda_d \geq 0$ and orthonormal eigenvectors $\bm{e}_1, \ldots, \bm{e}_d$ will play a central role in the speciation analysis of \Cref{sec:shocks}.

\subsection{The Cole--Hopf transformation}\label{sec:prelim:cole-hopf}

We recall the classical result that connects the heat equation to the Burgers equation~\citep{hopf1950partial, cole1951quasi}.

\begin{proposition}[Cole--Hopf; \citealp{hopf1950partial,cole1951quasi}]\label{prop:cole-hopf}
Let $\varphi(x,\tau)$ be a positive smooth solution of the heat equation $\varphi_\tau = \nu\,\varphi_{xx}$ in one spatial dimension.
Define
\begin{equation}\label{eq:cole-hopf-def}
  u(x,\tau) = -2\nu\,\frac{\partial_x \varphi}{\varphi} = -2\nu\,\partial_x \log \varphi.
\end{equation}
Then $u$ satisfies the viscous Burgers equation:
\begin{equation}\label{eq:burgers-standard}
  \frac{\partial u}{\partial \tau} + u\,\frac{\partial u}{\partial x} = \nu\,\frac{\partial^2 u}{\partial x^2}.
\end{equation}
\end{proposition}

The Burgers equation~\eqref{eq:burgers-standard} was introduced by \citet{burgers1948mathematical} as a one-dimensional model of turbulence.
The transformation~\eqref{eq:cole-hopf-def}, discovered independently by \citet{hopf1950partial} and \citet{cole1951quasi}, reduces it to the linear heat equation and provides explicit solutions for arbitrary initial data.
In the inviscid limit $\nu \to 0$, the equation $u_\tau + u\,u_x = 0$ develops gradient catastrophes in finite time---\emph{shock waves}---whose structure is governed by the Rankine--Hugoniot jump conditions~\citep{rankine1870thermodynamic, hugoniot1889propagation} and the Lax entropy condition~\citep{lax1957hyperbolic}.
The comprehensive treatment by \citet{whitham1974linear} and the modern PDE framework of \citet{evans2010partial} provide the mathematical foundations we employ.

\section{The Score--Burgers Correspondence}\label{sec:score-burgers}

The basic identification behind the rest of the paper is the following: the score function of a VE diffusion model satisfies a viscous Burgers equation.

\subsection{The one-dimensional score PDE}\label{sec:sb:1d}

\begin{theorem}[Score PDE]\label{thm:score-pde}
Let $p(x,\tau)$ be a positive smooth solution of the heat equation~\eqref{eq:heat} in one spatial dimension ($d=1$).
Then the score function $s(x,\tau) = \partial_x \log p(x,\tau)$ satisfies the nonlinear parabolic PDE
\begin{equation}\label{eq:score-pde}
  \frac{\partial s}{\partial \tau}
  = \frac{\partial^2 s}{\partial x^2}
    + 2\,s\,\frac{\partial s}{\partial x}.
\end{equation}
\end{theorem}

\begin{proof}
Since $s = p_x / p$, we have $p_x = s\,p$. Differentiating gives
\begin{align}
  p_{xx}
  &= (s\,p)_x
   = s_x\,p + s\,p_x
   = (s_x + s^2)\,p,\label{eq:pxx}\\[4pt]
  p_{xxx}
  &= \bigl[(s_x + s^2)\,p\bigr]_x
   = (s_{xx} + 2s\,s_x)\,p + (s_x + s^2)\,s\,p
   = (s_{xx} + 3s\,s_x + s^3)\,p.\label{eq:pxxx}
\end{align}
Differentiating $s = p_x / p$ with respect to $\tau$ and using the heat equation gives
\begin{equation}\label{eq:s-tau-quotient}
  \partial_\tau s
  = \frac{(\partial_\tau p_x)\,p - p_x\,(\partial_\tau p)}{p^2}
  = \frac{\partial_\tau p_x}{p} - s\,\frac{\partial_\tau p}{p}
  = \frac{p_{xxx}}{p} - s\,\frac{p_{xx}}{p}.
\end{equation}
Substituting~\eqref{eq:pxx} and~\eqref{eq:pxxx} yields
\begin{equation}\label{eq:s-tau-expand}
  \partial_\tau s
  = (s_{xx} + 3s\,s_x + s^3) - s\,(s_x + s^2)
  = s_{xx} + 2s\,s_x. \qedhere
\end{equation}
\end{proof}

In particular, the nonlinear score PDE is obtained exactly from the heat flow; no approximation enters at this stage.

\begin{remark}[Conservation form]\label{rem:conservation}
Equation~\eqref{eq:score-pde} can be written in divergence (conservation) form:
\begin{equation}\label{eq:score-conservation}
  \frac{\partial s}{\partial \tau}
  = \frac{\partial}{\partial x}\!\bigl(s_x + s^2\bigr).
\end{equation}
This form is natural for the analysis of weak solutions and will be used in the interfacial analysis of \Cref{sec:shocks}.
\end{remark}

\subsection{Identification with the Burgers equation}\label{sec:sb:identification}

\begin{theorem}[Score--Burgers correspondence]\label{thm:score-burgers}
Under the hypotheses of \Cref{thm:score-pde}, the function $u(x,\tau) = -2\,s(x,\tau)$ satisfies the viscous Burgers equation with unit viscosity:
\begin{equation}\label{eq:score-burgers}
  \frac{\partial u}{\partial \tau} + u\,\frac{\partial u}{\partial x}
  = \frac{\partial^2 u}{\partial x^2}.
\end{equation}
\end{theorem}

\begin{proof}
From $u = -2s$, we have $s = -u/2$, $s_x = -u_x/2$, $s_{xx} = -u_{xx}/2$, and $s_\tau = -u_\tau/2$.
Substituting into~\eqref{eq:score-pde}:
\[
  -\frac{u_\tau}{2}
  = -\frac{u_{xx}}{2} + 2\Bigl(-\frac{u}{2}\Bigr)\Bigl(-\frac{u_x}{2}\Bigr)
  = -\frac{u_{xx}}{2} + \frac{u\,u_x}{2}.
\]
Multiplying by $-2$: $u_\tau = u_{xx} - u\,u_x$, which is~\eqref{eq:score-burgers}.
\end{proof}

So every one-dimensional VE score field can be read as a Burgers velocity after the simple rescaling $u=-2s$.

\begin{remark}[Exactness]\label{rem:exactness}
This identification is an identity rather than an approximation. It can be read off directly from \Cref{prop:cole-hopf} by setting $\varphi = p_\tau$, which solves the heat equation~\eqref{eq:heat} in $\tau$-time with $\nu = 1$, and observing that the Cole--Hopf variable~\eqref{eq:cole-hopf-def} becomes $u = -2\,\partial_x \log p_\tau = -2\,s$. The observation that the ``score Fokker--Planck equation'' of \citet{lai2023fp} is the Burgers equation does not appear explicitly there.
\end{remark}

\subsection{Physical-time formulation}\label{sec:sb:physical}

Reverting from $\tau$-time to the physical time $t$ of the VE-SDE~\eqref{eq:ve-sde} using $\partial_\tau = \nu(t)^{-1}\partial_t$ yields:

\begin{corollary}[Score PDE in physical time]\label{cor:physical-time}
In the physical time $t$ of the VE-SDE with diffusion coefficient $g(t)$, the score satisfies
\begin{equation}\label{eq:score-pde-physical}
  \frac{\partial s}{\partial t}
  = \nu(t)\!\left(\frac{\partial^2 s}{\partial x^2} + 2\,s\,\frac{\partial s}{\partial x}\right),
  \qquad \nu(t) = \frac{g(t)^2}{2},
\end{equation}
and $u = -2s$ satisfies the viscous Burgers equation with time-dependent viscosity:
\begin{equation}\label{eq:burgers-physical}
  \frac{\partial u}{\partial t} + \nu(t)\,u\,\frac{\partial u}{\partial x}
  = \nu(t)\,\frac{\partial^2 u}{\partial x^2}.
\end{equation}
\end{corollary}

The appearance of $\nu(t)$ as a time-dependent viscosity means that the Burgers dynamics is ``fast'' when the noise injection rate $g(t)$ is large and ``slow'' when $g(t)$ is small.
This has direct implications for the noise schedule design: the inviscid limit (where shocks form) is approached whenever $\nu(t) \to 0$, i.e., at the beginning of the forward process and---critically---at the end of the reverse (generative) process when noise is nearly removed.

\subsection{Connection to the score Fokker--Planck equation}\label{sec:sb:connection}

\citet{lai2023fp} derived the PDE governing the time evolution of the score by differentiating the Fokker--Planck equation~\eqref{eq:fpe}.
They termed the result the ``score Fokker--Planck equation'' (score FPE) and used it as a training regularizer, showing empirically that enforcing it improves both log-likelihood and the conservativity of the learned score.

In the VE setting, their score FPE is simply our~\eqref{eq:score-pde}.
To see this, note that the general form given by \citet[Eq.~(8)]{lai2023fp} reduces, for the VE-SDE with $\bm{f} = \bm{0}$ and scalar $g(t)$, to
\[
  \partial_t s_i = \nu(t)\bigl[\Delta s_i + 2\,s_j\,\partial_j s_i\bigr],
\]
which is the $d$-dimensional generalization of~\eqref{eq:score-pde-physical} (the vector Burgers system; see \Cref{sec:multidim}).
The connection to the Burgers equation appears not to have been noted in their work.

\subsection{Informal summary}

Up to the fixed rescaling $u=-2s$, the score of a VE diffusion model is a Burgers velocity field. Reverse time lowers the noise, so the effective viscosity drops and the boundary layers sharpen. The next section works this out first in the symmetric binary case.

\section{Interfacial Structure and Speciation}\label{sec:shocks}

The symmetric binary Gaussian mixture is the cleanest place to see the mechanism. There the score profile, the normal Hessian, and the interfacial width are all explicit. The local theorem comes out of that calculation, and only afterwards do we return to the mixture-specific consequences.

\subsection{Exact score for a symmetric two-component mixture}\label{sec:shocks:exact}

Consider the symmetric binary Gaussian mixture in one dimension:
\begin{equation}\label{eq:binary-gmm}
  p_0(x) = \tfrac{1}{2}\,\mathcal{N}(x;\,-a,\,\sigma_0^2) + \tfrac{1}{2}\,\mathcal{N}(x;\,a,\,\sigma_0^2),
\end{equation}
with mode half-separation $a > 0$ and component variance $\sigma_0^2 > 0$.
Under the VE forward process at diffusion time $\tau$, the noised density is
\begin{equation}\label{eq:binary-noised}
  p_\tau(x) = \tfrac{1}{2}\,\mathcal{N}(x;\,-a,\,\sigma_\tau^2) + \tfrac{1}{2}\,\mathcal{N}(x;\,a,\,\sigma_\tau^2),
  \qquad \sigma_\tau^2 = \sigma_0^2 + 2\tau.
\end{equation}

\begin{proposition}[Exact score formula]\label{prop:exact-score}
The score of the noised density~\eqref{eq:binary-noised} is
\begin{equation}\label{eq:exact-score}
  s(x,\tau)
  = -\frac{x}{\sigma_\tau^2}
    + \frac{a}{\sigma_\tau^2}\,\tanh\!\left(\frac{a\,x}{\sigma_\tau^2}\right).
\end{equation}
\end{proposition}

\begin{proof}
Let $\varphi_\pm(x) = \mathcal{N}(x;\,\pm a,\,\sigma_\tau^2)$.
Then $\partial_x \varphi_\pm = -(x \mp a)\,\sigma_\tau^{-2}\,\varphi_\pm$, and the score is
\[
  s = \frac{\tfrac{1}{2}\partial_x \varphi_- + \tfrac{1}{2}\partial_x \varphi_+}
          {\tfrac{1}{2}\varphi_- + \tfrac{1}{2}\varphi_+}
    = -\frac{x}{\sigma_\tau^2}
      + \frac{a}{\sigma_\tau^2}\cdot\frac{\varphi_+ - \varphi_-}{\varphi_+ + \varphi_-}.
\]
We compute the ratio.
Since $\varphi_\pm \propto \exp\!\bigl(-(x \mp a)^2/(2\sigma_\tau^2)\bigr)$,
\[
  \frac{\varphi_+}{\varphi_-}
  = \exp\!\left(\frac{(x+a)^2 - (x-a)^2}{2\sigma_\tau^2}\right)
  = \exp\!\left(\frac{4ax}{2\sigma_\tau^2}\right)
  = \exp\!\left(\frac{2ax}{\sigma_\tau^2}\right),
\]
where we expanded $(x+a)^2 - (x-a)^2 = 4ax$.
Therefore,
\[
  \frac{\varphi_+ - \varphi_-}{\varphi_+ + \varphi_-}
  = \frac{e^{2ax/\sigma_\tau^2} - 1}{e^{2ax/\sigma_\tau^2} + 1}
  = \tanh\!\left(\frac{ax}{\sigma_\tau^2}\right).\qedhere
\]
\end{proof}

The score profiles and their Burgers transform are shown in \Cref{fig:score-burgers}. Across diffusion times, the score develops the narrow inter-mode transition whose background-subtracted form becomes the Burgers shock analyzed below.

\begin{figure}[t]
\centering
\includegraphics[width=\linewidth]{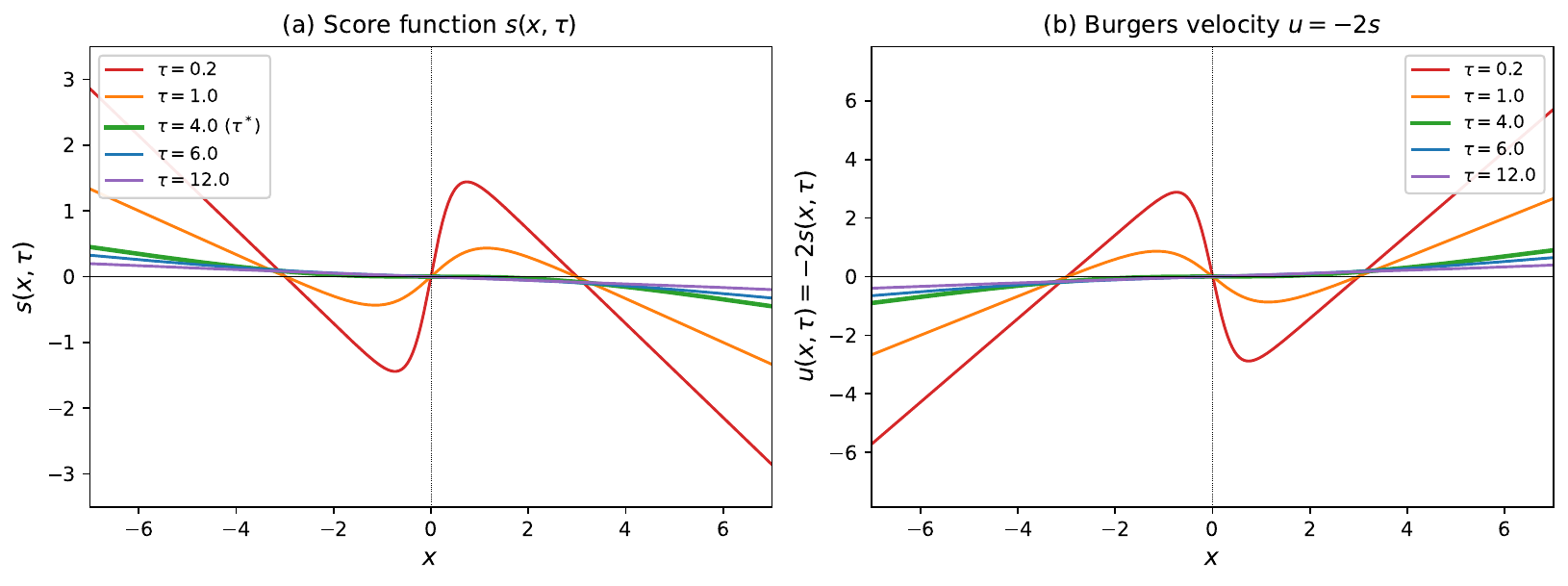}
\caption{Symmetric binary Gaussian mixture at several diffusion times. Panel~(a) traces the exact score $s(x,\tau)$; the central transition sharpens as $\tau \downarrow 0$, and the critical time $\tau^\ast = 4.0$ is marked. Panel~(b) displays the Burgers variable $u=-2s$ for the same slices. Its linear background remains visible; subtracting $2x/\sigma_\tau^2$ isolates the $\tanh$ layer described in Proposition~\ref{prop:shock-profile}.}
\label{fig:score-burgers}
\end{figure}

\subsection{The midpoint derivative of the score and the critical time}\label{sec:shocks:hessian}

The quantity $s_x(0,\tau)=\partial_x^2 \log p_\tau(0)$ is the midpoint derivative of the score, equivalently the one-dimensional Hessian of $\log p_\tau$ at the mode boundary. Its behavior governs the transition from unimodal to bimodal structure.

\begin{proposition}[Midpoint derivative of the score]\label{prop:sx-origin}
For the symmetric binary mixture~\eqref{eq:binary-gmm},
\begin{equation}\label{eq:sx-origin}
  s_x(0,\tau)
  = \frac{a^2 - \sigma_\tau^2}{\sigma_\tau^4}.
\end{equation}
In particular, $s_x(0,\tau) = 0$ if and only if $\sigma_\tau^2 = a^2$.
\end{proposition}

\begin{proof}
Differentiating~\eqref{eq:exact-score} with respect to $x$:
\[
  s_x(x,\tau)
  = -\frac{1}{\sigma_\tau^2}
    + \frac{a^2}{\sigma_\tau^4}\,\sech^2\!\left(\frac{ax}{\sigma_\tau^2}\right).
\]
At $x = 0$: $\sech^2(0) = 1$, giving $s_x(0,\tau) = -\sigma_\tau^{-2} + a^2\,\sigma_\tau^{-4} = (a^2 - \sigma_\tau^2)/\sigma_\tau^4$.
\end{proof}

The sign of $s_x(0,\tau)$ determines the local shape of $\log p_\tau$ at the midpoint:
\begin{itemize}[nosep]
\item If $\sigma_\tau^2 > a^2$ (i.e., $\tau > \tau^\ast$): $s_x(0,\tau) < 0$, so $x = 0$ is a local maximum of $\log p_\tau$---the density appears \emph{unimodal}.
\item If $\sigma_\tau^2 < a^2$ (i.e., $\tau < \tau^\ast$): $s_x(0,\tau) > 0$, so $x = 0$ is a local minimum of $\log p_\tau$---the density is \emph{bimodal}.
\end{itemize}
The transition occurs at the \emph{critical diffusion time}:
\begin{definition}[Speciation time]\label{def:tau-star}
The speciation time for the symmetric binary mixture~\eqref{eq:binary-gmm} is
\begin{equation}\label{eq:tau-star}
  \tau^\ast = \frac{a^2 - \sigma_0^2}{2},
\end{equation}
assuming $a > \sigma_0$ (modes separated by more than one standard deviation).
At $\tau = \tau^\ast$, $\sigma_{\tau^\ast}^2 = \sigma_0^2 + 2\tau^\ast = a^2$.
\end{definition}

In the reverse generative process, which traverses diffusion time from $\tau_T \gg 1$ down to $\tau = 0$, the speciation time $\tau^\ast$ marks the moment at which the unimodal score field bifurcates: a single attractor at $x = 0$ splits into two attractors near $x = \pm a$.
This is the \emph{speciation transition} of \citet{biroli2024dynamical}, who identified it (in a high-dimensional mean-field framework) as a symmetry-breaking phase transition between their Regime~I (noise-dominated) and Regime~II (class-committed).
This change of local geometry is shown in \Cref{fig:speciation-width}, where the midpoint derivative crosses zero at $\tau^\ast$ and the associated interfacial width varies linearly with diffusion time.

\begin{figure}[t]
\centering
\includegraphics[width=\linewidth]{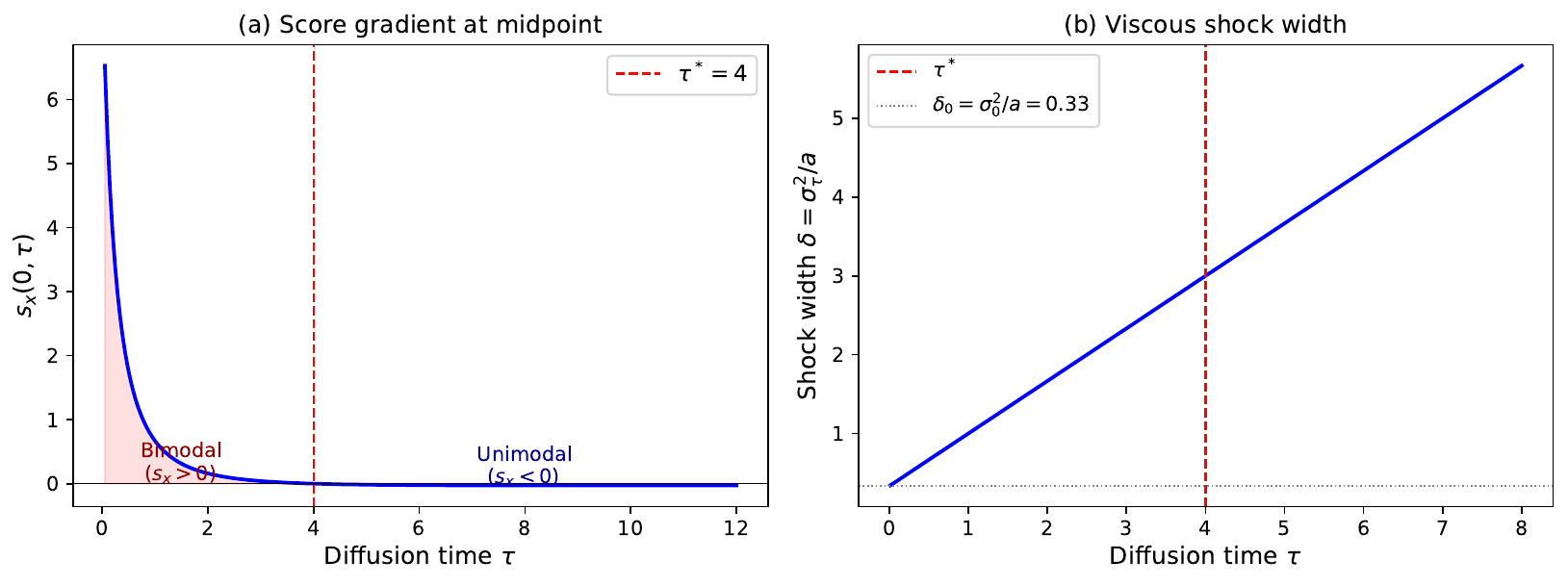}
\caption{Two simple diagnostics for the symmetric binary mixture. Panel~(a) tracks the midpoint derivative $s_x(0,\tau)=\partial_x^2\log p_\tau(0)$, whose zero marks the exact speciation time $\tau^\ast = 4.0$. Panel~(b) records the interfacial width $\delta(\tau)=\sigma_\tau^2/a$, linear in $\tau$ and smallest at $\tau=0$.}
\label{fig:speciation-width}
\end{figure}

\subsection{The background-subtracted interfacial shock profile}\label{sec:shocks:profile}

For the symmetric binary mixture, the inter-mode layer separates cleanly from the ambient Gaussian drift. After subtracting that linear background term, the remaining profile is the classical viscous Burgers shock.

\begin{proposition}[Background-subtracted interfacial profile]\label{prop:shock-profile}
Define the background-subtracted score
\begin{equation}\label{eq:background-subtracted-score}
  \tilde{s}(x,\tau)
  \coloneqq s(x,\tau) + \frac{x}{\sigma_\tau^2}
  = \frac{a}{\sigma_\tau^2}\,\tanh\!\left(\frac{a\,x}{\sigma_\tau^2}\right),
\end{equation}
and the corresponding Burgers variable $\tilde{u} = -2\tilde{s}$.
Then $\tilde{u}$ has left and right asymptotic states
\begin{equation}\label{eq:shock-states}
  \tilde{u}_L = \frac{2a}{\sigma_\tau^2}, \qquad \tilde{u}_R = -\frac{2a}{\sigma_\tau^2},
\end{equation}
and the $\tanh$ transition between them has width
\begin{equation}\label{eq:shock-width}
  \delta(\tau) = \frac{\sigma_\tau^2}{a}.
\end{equation}
Thus the inter-mode layer is exactly the classical viscous Burgers shock after subtraction of the linear Gaussian background drift.
\end{proposition}

\begin{proof}
Equation~\eqref{eq:background-subtracted-score} follows immediately from the exact score formula~\eqref{eq:exact-score}. Therefore
\[
  \tilde{u}(x,\tau)
  = -\frac{2a}{\sigma_\tau^2}\,\tanh\!\left(\frac{a\,x}{\sigma_\tau^2}\right).
\]
The classical steady viscous Burgers shock connecting states $u_L > u_R$ with viscosity $\nu$ is~\citep[Ch.~4]{whitham1974linear}:
\begin{equation}\label{eq:classical-shock}
  u(x) = \frac{u_L + u_R}{2} - \frac{u_L - u_R}{2}\,\tanh\!\left(\frac{(u_L - u_R)\,x}{4\nu}\right),
\end{equation}
with shock width $\delta = 4\nu/(u_L - u_R)$.
Comparing with $\tilde{u}$ at viscosity $\nu = 1$ gives $\tilde{u}_L = 2a/\sigma_\tau^2$, $\tilde{u}_R = -2a/\sigma_\tau^2$, and
\[
  \delta = \frac{4}{4a/\sigma_\tau^2} = \frac{\sigma_\tau^2}{a}.
\]
\end{proof}

\begin{remark}[Interfacial sharpening]\label{rem:sharpening}
As the generative process proceeds ($\tau$ decreases toward $0$), the interfacial width $\delta$ shrinks to $\sigma_0^2/a$. The midpoint derivative of the score is
\[
  s_x(0,\tau)=\frac{a^2-\sigma_\tau^2}{\sigma_\tau^4},
\]
and approaches $(a^2-\sigma_0^2)/\sigma_0^4$ as $\tau\to 0$; it diverges only in the point-mass limit $\sigma_0\to 0$. So for any finite variance the layer stays smooth, just increasingly narrow. The actual jump appears only in the inviscid point-mass limit.
\end{remark}

\subsection{The exact local binary-boundary theorem}\label{sec:shocks:local}

The symmetric Gaussian formulas above reveal the mechanism transparently, but the underlying $\tanh$ layer is not a Gaussian artifact.
It is an exact algebraic consequence of binary competition between two positive heat contributions.
A canonical choice is to partition the initial density into two attraction basins $\Omega_1,\Omega_2$ and define
\[
  p_\tau^{(k)}(\bm{x}) = \int_{\Omega_k} p_0(\bm{y})\,G_\tau(\bm{x}-\bm{y})\,d\bm{y},
  \qquad k=1,2,
\]
so that each $p_\tau^{(k)}$ satisfies the heat equation by linearity.
The theorem below, however, requires only positivity and separate heat evolution.

\begin{theorem}[Exact binary decomposition]\label{thm:local-tanh}
Let $p_\tau = p_\tau^{(1)} + p_\tau^{(2)}$ on $\R^d$, where $p_\tau^{(1)},p_\tau^{(2)}>0$ are smooth and each satisfies $\partial_\tau p_\tau^{(k)} = \Delta p_\tau^{(k)}$.
Define
\begin{equation}\label{eq:local-phi-def}
  \phi(\bm{x},\tau) = \log \frac{p_\tau^{(1)}(\bm{x})}{p_\tau^{(2)}(\bm{x})},
  \qquad
  \bar{\bm{s}}(\bm{x},\tau) = \tfrac{1}{2}\bigl(\nabla \log p_\tau^{(1)}(\bm{x}) + \nabla \log p_\tau^{(2)}(\bm{x})\bigr).
\end{equation}
Then the full score $\bm{s} = \nabla \log p_\tau$ satisfies the exact identity
\begin{equation}\label{eq:local-tanh}
  \bm{s}(\bm{x},\tau)
  = \bar{\bm{s}}(\bm{x},\tau) + \tfrac{1}{2}\tanh\!\left(\frac{\phi(\bm{x},\tau)}{2}\right)\nabla\phi(\bm{x},\tau).
\end{equation}
\end{theorem}

\begin{proof}
Write $\bm{s}_k = \nabla\log p_\tau^{(k)}$ and $R = p_\tau^{(1)}/p_\tau^{(2)} = e^{\phi}$.
Then
\[
  \bm{s}
  = \frac{p_\tau^{(1)}\bm{s}_1 + p_\tau^{(2)}\bm{s}_2}{p_\tau^{(1)} + p_\tau^{(2)}}
  = \frac{R}{1+R}\,\bm{s}_1 + \frac{1}{1+R}\,\bm{s}_2.
\]
Since $\bm{s}_1 = \bar{\bm{s}} + \tfrac{1}{2}\nabla\phi$ and $\bm{s}_2 = \bar{\bm{s}} - \tfrac{1}{2}\nabla\phi$, we obtain
\[
  \bm{s}
  = \bar{\bm{s}} + \frac{R-1}{2(R+1)}\nabla\phi
  = \bar{\bm{s}} + \tfrac{1}{2}\tanh\!\left(\frac{\phi}{2}\right)\nabla\phi,
\]
using $(e^{\phi}-1)/(e^{\phi}+1)=\tanh(\phi/2)$.
\end{proof}

\begin{proposition}[Log-ratio advection--diffusion]\label{prop:local-phi-pde}
Under the hypotheses of \Cref{thm:local-tanh}, the log-ratio $\phi$ satisfies
\begin{equation}\label{eq:local-phi-pde}
  \partial_\tau \phi = \Delta \phi + 2\,\bar{\bm{s}}\cdot\nabla\phi.
\end{equation}
\end{proposition}

\begin{proof}
For each $k$, positivity and the heat equation give
\[
  \partial_\tau \log p_\tau^{(k)}
  = \frac{\Delta p_\tau^{(k)}}{p_\tau^{(k)}}
  = \Delta \log p_\tau^{(k)} + |\nabla\log p_\tau^{(k)}|^2.
\]
Subtracting the identities for $k=1$ and $k=2$ yields
\[
  \partial_\tau \phi
  = \Delta\phi + |\bm{s}_1|^2 - |\bm{s}_2|^2
  = \Delta\phi + (\bm{s}_1+\bm{s}_2)\cdot(\bm{s}_1-\bm{s}_2)
  = \Delta\phi + 2\,\bar{\bm{s}}\cdot\nabla\phi.
\]
\end{proof}

\begin{theorem}[Local boundary-normal reduction and exact speciation criterion]\label{thm:local-speciation}
Assume the binary boundary
\begin{equation}\label{eq:binary-boundary}
  \Gamma_\tau = \{\bm{x}\in\R^d : \phi(\bm{x},\tau)=0\}
\end{equation}
is regular, i.e., $\nabla\phi\neq 0$ on $\Gamma_\tau$.
Fix $\bm{x}_\Gamma\in\Gamma_\tau$, let
\begin{equation}\label{eq:kappa-def}
  \hat{\bm{n}} = \frac{\nabla\phi}{|\nabla\phi|}\Big|_{\bm{x}_\Gamma},
  \qquad
  \kappa = |\nabla\phi(\bm{x}_\Gamma,\tau)|,
\end{equation}
and use boundary-normal coordinates $(n,\bm{y})$ with signed distance $n$ in the $\hat{\bm{n}}$ direction.
Then, as $n\to 0$,
\begin{equation}\label{eq:local-normal-profile}
  (\bm{s}-\bar{\bm{s}})\cdot\hat{\bm{n}}
  = \tfrac{1}{2}\kappa\,\tanh\!\left(\frac{\kappa n}{2}\right) + O(n),
\end{equation}
and exactly on the boundary,
\begin{equation}\label{eq:local-normal-criterion}
  \partial_n s_n\big|_{\Gamma_\tau}
  = \partial_n \bar s_n\big|_{\Gamma_\tau} + \frac{\kappa^2}{4},
  \qquad s_n = \bm{s}\cdot\hat{\bm{n}},\ \bar s_n = \bar{\bm{s}}\cdot\hat{\bm{n}}.
\end{equation}
Thus the boundary-normal slice is locally bimodal at $\bm{x}_\Gamma$ if and only if
\begin{equation}\label{eq:local-bimodal-criterion}
  \partial_n \bar s_n\big|_{\Gamma_\tau} + \frac{\kappa^2}{4} > 0.
\end{equation}
\end{theorem}

\begin{proof}
Because $\phi=0$ on $\Gamma_\tau$ and $\nabla\phi\neq 0$ there, boundary-normal coordinates give
\[
  \phi(n,\bm{y},\tau) = \kappa n + O(n^2),
  \qquad
  \partial_n\phi(n,\bm{y},\tau) = \kappa + O(n).
\]
Taking the normal component of \eqref{eq:local-tanh} yields
\[
  s_n - \bar s_n
  = \tfrac{1}{2}\tanh\!\left(\frac{\phi}{2}\right)\partial_n\phi
  = \tfrac{1}{2}\tanh\!\left(\frac{\kappa n}{2} + O(n^2)\right)\bigl(\kappa + O(n)\bigr),
\]
which is \eqref{eq:local-normal-profile}.
For the boundary derivative, differentiate the identity
\[
  s_n = \bar s_n + \tfrac{1}{2}\tanh\!\left(\frac{\phi}{2}\right)\partial_n\phi
\]
along the normal coordinate.
At $n=0$ one has $\phi=0$, hence $\tanh(0)=0$ and $\operatorname{sech}^2(0)=1$, so
\[
  \partial_n s_n\big|_{\Gamma_\tau}
  = \partial_n \bar s_n\big|_{\Gamma_\tau}
    + \tfrac{1}{2}\left(\frac{\partial_n\phi}{2}\right)\partial_n\phi\Big|_{\Gamma_\tau}
  = \partial_n \bar s_n\big|_{\Gamma_\tau} + \frac{\kappa^2}{4}.
\]
On a one-dimensional normal slice, local bimodality is equivalent to the second derivative of $\log p_\tau$ being positive at the boundary point, i.e., to $\partial_n s_n>0$.
This gives \eqref{eq:local-bimodal-criterion}.
\end{proof}

\begin{remark}[What is universal, and what is model-specific]\label{rem:local-universal}
\Cref{thm:local-tanh,thm:local-speciation} separate what is universal from what depends on the model. The $\tanh$ layer and the positive term $\kappa^2/4$ are universal for binary competition. By contrast, the actual objects $\phi$, $\bar{\bm{s}}$, and therefore $\kappa$ depend on the model. It is also worth keeping in mind that \Cref{prop:local-phi-pde} is linear in $\phi$: the sharp score layer comes from the nonlinear map $\phi\mapsto \tfrac{1}{2}\tanh(\phi/2)\nabla\phi$, not from shock formation in $\phi$ itself. In the symmetric Gaussian case one has $\phi(x,\tau)=2ax/\sigma_\tau^2$ and $\bar s(x,\tau)=-x/\sigma_\tau^2$, so the general statement reduces to \Cref{prop:shock-profile,thm:shock-speciation}.
\end{remark}

\begin{proposition}[Error from non-binary competitors]\label{prop:binary-remainder}
Suppose the true density admits the decomposition
\begin{equation}\label{eq:binary-remainder}
  p_\tau = p_\tau^{(1)} + p_\tau^{(2)} + p_\tau^{(\mathrm{rem})},
  \qquad
  r \coloneqq \frac{p_\tau^{(\mathrm{rem})}}{p_\tau^{(1)} + p_\tau^{(2)}}.
\end{equation}
Let $\bm{s}^{(\mathrm{bin})}$ denote the score built from $p_\tau^{(1)}+p_\tau^{(2)}$ via \Cref{thm:local-tanh}.
Then
\begin{equation}\label{eq:binary-remainder-score}
  \bm{s} - \bm{s}^{(\mathrm{bin})} = \nabla\log(1+r),
\end{equation}
and on a boundary-normal slice,
\begin{equation}\label{eq:binary-remainder-hessian}
  \partial_n s_n
  = \partial_n s_n^{(\mathrm{bin})} + \partial_n^2\log(1+r).
\end{equation}
Thus the exact binary criterion \eqref{eq:local-normal-criterion} remains accurate whenever $r$, $\partial_n r$, and $\partial_n^2 r$ are small; for well-separated competing modes, these corrections are exponentially small in the distance to the nearest non-competing mode measured in units of $\sqrt{\tau}$.
\end{proposition}

\begin{proof}
Since $p_\tau = (p_\tau^{(1)} + p_\tau^{(2)})(1+r)$,
\[
  \log p_\tau = \log\bigl(p_\tau^{(1)} + p_\tau^{(2)}\bigr) + \log(1+r),
\]
so differentiating once and twice along the normal direction gives \eqref{eq:binary-remainder-score} and \eqref{eq:binary-remainder-hessian}.
The exponential smallness is the standard heat-kernel suppression of a farther mode relative to the two competing ones.
\end{proof}

\subsection{The Gaussian specialization and the spectral threshold}\label{sec:shocks:characteristics}

For the symmetric Gaussian model there is nothing subtle left: the exact local criterion reduces to the midpoint-derivative test, and that is the same condition picked out by the spectral criterion of \citet{biroli2024dynamical}.

\begin{theorem}[Gaussian specialization: speciation criterion = spectral threshold]\label{thm:shock-speciation}
For the symmetric binary Gaussian mixture~\eqref{eq:binary-gmm}, the exact local criterion of \Cref{thm:local-speciation} reduces to the midpoint-derivative criterion, and the following two quantities coincide:
\begin{enumerate}[label=\textup{(\roman*)},nosep]
\item The critical diffusion time $\tau^\ast$ at which $s_x(0,\tau^\ast) = 0$ (equivalently, the one-dimensional Hessian of $\log p_{\tau^\ast}$ vanishes at the mode boundary).
\item The speciation time of \citet{biroli2024dynamical}, defined as the time at which the largest non-trivial eigenvalue of the noised data covariance equals the noise variance.
\end{enumerate}
In Burgers terms, $\tau^\ast$ is the threshold at which the inter-mode layer changes from a single-attractor profile to a split, shock-like interface.
\end{theorem}

\begin{proof}
For the one-dimensional symmetric binary mixture, the between-class covariance~\eqref{eq:W-def} is $W = w_1 \nu_1^2 + w_2 \nu_2^2 = \tfrac{1}{2}a^2 + \tfrac{1}{2}a^2 = a^2$ (a scalar, since $d=1$, with $\nu_1 = -a$, $\nu_2 = a$).
The spectral criterion of \citet{biroli2024dynamical} states that speciation occurs when the signal-to-noise ratio---the ratio of the largest between-class eigenvalue to the noise variance---crosses unity:
\begin{equation}\label{eq:biroli-criterion}
  \frac{\lambda_1^{(W)}}{\sigma_\tau^2} = 1
  \quad\Longleftrightarrow\quad
  \sigma_\tau^2 = \lambda_1^{(W)} = a^2
  \quad\Longleftrightarrow\quad
  \tau = \frac{a^2 - \sigma_0^2}{2} = \tau^\ast.
\end{equation}
This is identical to~\eqref{eq:tau-star}.
\end{proof}

\begin{remark}[Interpretation]\label{rem:interpretation}
In the symmetric Gaussian case, \Cref{thm:shock-speciation} agrees with two standard ways of locating the transition: the midpoint derivative changes sign, and the spectral signal-to-noise ratio crosses one. The Burgers language is not doing a separate characteristic calculation here. What it does give is a more concrete picture of the interface itself: its profile, its width, and the Rankine--Hugoniot motion law.
\end{remark}

\subsection{The Rankine--Hugoniot condition for asymmetric mixtures}\label{sec:shocks:rh}

For unequal-weight mixtures ($w_1 \neq w_2$), the shock is located at a point $x_s(\tau) \neq 0$ that drifts as $\tau$ changes.
Its motion is governed by the Rankine--Hugoniot condition~\citep{rankine1870thermodynamic, hugoniot1889propagation, lax1957hyperbolic}:

\begin{proposition}[Decision boundary dynamics]\label{prop:rh}
In the inviscid limit, the location $x_s(\tau)$ of the score shock between two modes satisfies
\begin{equation}\label{eq:rh}
  \frac{dx_s}{d\tau} = -\bigl(s_L(\tau) + s_R(\tau)\bigr),
\end{equation}
where $s_L$ and $s_R$ are the score values on the left and right sides of the shock.
\end{proposition}

\begin{proof}
For the inviscid score equation in conservation form~\eqref{eq:score-conservation}, the standard Rankine--Hugoniot jump condition~\citep[Thm.~3.4.1]{evans2010partial} gives the shock speed as
\[
  \dot{x}_s = \frac{[s_x + s^2]}{[s]} = \frac{[s^2]}{[s]} = s_L + s_R,
\]
where $[\cdot]$ denotes the jump across the shock and we used $[s_x] = 0$ in the distributional sense for the shock solution.
(For the flux $f(s) = s^2$, the Rankine--Hugoniot speed is $(f(s_R) - f(s_L))/(s_R - s_L) = s_L + s_R$.)
The minus sign in~\eqref{eq:rh} arises from reversing the time orientation when tracing the generative (reverse) process.
For $w_1 = w_2$, symmetry gives $s_L = -s_R$ at $x = 0$, hence $\dot{x}_s = 0$: the shock is stationary.
For $w_1 \neq w_2$, the boundary drifts toward the minority component.
\end{proof}

\subsection{The Lax entropy condition and mode stability}\label{sec:shocks:entropy}

In one spatial dimension, the physical relevance of Burgers shocks is determined by the Lax entropy condition~\citep{lax1957hyperbolic}: a shock with left state $u_L$ and right state $u_R$ is admissible (entropy-satisfying) if and only if $u_L > u_R$.
Translating to the score ($u = -2s$), this becomes $s_L < s_R$: the score must jump from a lower value (pointing toward the left mode) to a higher value (pointing toward the right mode) as one crosses the boundary from left to right.

\begin{proposition}[Scalar entropy admissibility on boundary slices]\label{prop:entropy}
For any Gaussian mixture with well-separated modes, the scalar score profile along a one-dimensional normal slice through an inter-mode boundary satisfies the Lax entropy condition.
\end{proposition}

\begin{proof}
Between two adjacent modes with means $\mu_j < \mu_k$, the score far to the left of the boundary is dominated by component $j$: $s \approx -(x - \mu_j)/\sigma_\tau^2 < 0$ for $x$ between the modes (since $x > \mu_j$).
Far to the right, $s \approx -(x - \mu_k)/\sigma_\tau^2 > 0$ (since $x < \mu_k$).
Hence $s_L < 0 < s_R$, and the Lax condition $s_L < s_R$ is satisfied.
\end{proof}

A learned score network that violates the scalar entropy condition on such a slice would correspond to an ``entropy-violating weak solution'' of the Burgers equation~\citep{lax1957hyperbolic}---a non-physical shock that can cause spurious mode creation or mode collapse in the generated distribution.
This provides a useful diagnostic: one can check the Lax condition along estimated boundary-normal slices to detect pathological score network behavior.

For completeness, \Cref{fig:pde-verification} verifies the score PDE and Burgers equation directly by finite-difference residuals; the errors remain at machine precision throughout the tested diffusion times.

\begin{figure}[t]
\centering
\includegraphics[width=\linewidth]{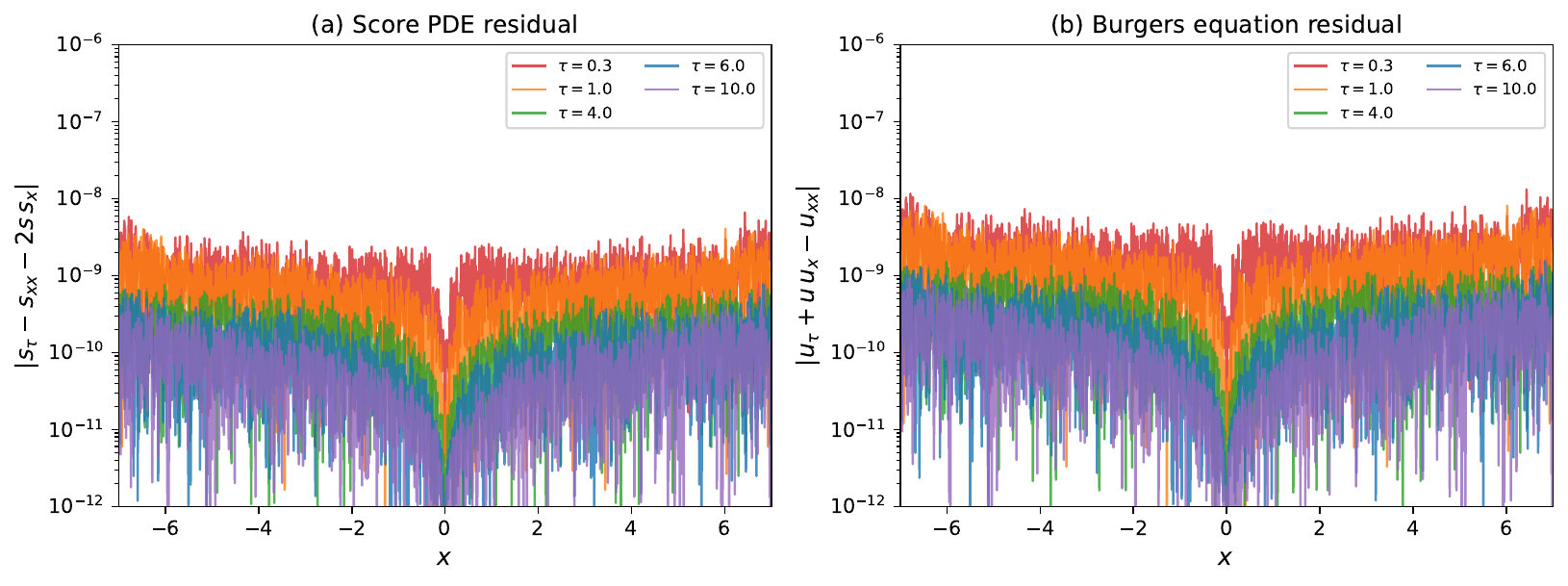}
\caption{Residual checks for the PDE identities. Panel~(a) reports $|s_\tau - s_{xx} - 2s s_x|$ at several diffusion times. Panel~(b) reports the Burgers residual $|u_\tau + uu_x - u_{xx}|$ after the change of variables $u=-2s$. Both remain below $10^{-8}$ on the tested grid.}
\label{fig:pde-verification}
\end{figure}

\section{Error Amplification at Score Shocks}\label{sec:amplification}

With the interfacial structure now established---exactly in the symmetric Gaussian model and locally in general through the binary-boundary theorem---we turn to its main dynamical consequence for generation.
Errors in the score are amplified when trajectories traverse the boundary layer, and in the symmetric Gaussian case this amplification can be computed in closed form as a function of the signal-to-noise ratio.
The resulting growth factor and the associated trajectory bifurcation are displayed in \Cref{fig:amplification-trajectories}.

\subsection{Trajectory divergence near the interfacial layer}\label{sec:amp:trajectory}

The probability flow ODE~\eqref{eq:prob-flow} for the VE-SDE (in $\tau$-time, running backward from $\tau_T$ to $0$) reduces to
\begin{equation}\label{eq:prob-flow-tau}
  \frac{dx}{d\tau} = -s(x,\tau), \qquad \tau \text{ decreasing from } \tau_T \text{ to } 0.
\end{equation}
Linearizing around a trajectory $x(\tau)$, a small perturbation $\xi(\tau) = \delta x(\tau)$ satisfies
\begin{equation}\label{eq:linearized-flow}
  \frac{d\xi}{d\tau} = -s_x(x(\tau),\tau)\,\xi.
\end{equation}
For a trajectory passing through the shock center at $x \approx 0$, the local growth rate is $-s_x(0,\tau)$.

\begin{proposition}[Lyapunov exponent at the shock]\label{prop:lyapunov}
For the symmetric binary mixture~\eqref{eq:binary-gmm} with $\tau < \tau^\ast$, the reverse-time Lyapunov exponent at the mode boundary is
\begin{equation}\label{eq:lyapunov}
  \lambda(\tau) = s_x(0,\tau) = \frac{a^2 - \sigma_\tau^2}{\sigma_\tau^4} > 0.
\end{equation}
Nearby generative trajectories diverge locally at rate $\lambda(\tau)$ during the reverse process.
\end{proposition}

\begin{proof}
In the reverse direction ($\tau$ decreasing), the perturbation equation~\eqref{eq:linearized-flow} becomes $d\xi/d\sigma = s_x(0,\tau)\,\xi$ with $\sigma = \tau_T - \tau$ increasing.
Since $s_x(0,\tau) > 0$ for $\tau < \tau^\ast$ by \Cref{prop:sx-origin}, perturbations grow.
This local trajectory divergence is the hallmark of the speciation bifurcation: infinitesimally close initial conditions lead to macroscopically different modes~\citep{raya2023spontaneous, biroli2024dynamical}.
\end{proof}

\subsection{The Gr\"onwall bound with score error}\label{sec:amp:gronwall}

We next examine how score-estimation errors are amplified near the interfacial layer. The first result is a general trajectory-stability bound for the probability flow ODE; the second specializes it to the symmetric binary mixture and gives a closed-form exponent.
Let $\hat{s}(x,\tau)$ be a learned score approximation with pointwise error bounded by $\varepsilon(\tau)$.

\begin{theorem}[Trajectory error amplification]\label{thm:trajectory-error}
Let $x(\tau)$ and $\hat{x}(\tau)$ be trajectories of the probability flow ODE~\eqref{eq:prob-flow-tau} driven by the true score $s$ and approximate score $\hat{s}$ respectively, starting from the same initial point $x(\tau_T) = \hat{x}(\tau_T)$.
Define the trajectory error $e(\tau) = |x(\tau) - \hat{x}(\tau)|$ and the uniform score error $\varepsilon_0 = \sup_\tau \|\hat{s}(\cdot,\tau) - s(\cdot,\tau)\|_{L^\infty}$.
Then for all $\tau \in [0,\tau_T]$:
\begin{equation}\label{eq:gronwall-bound}
  |e(\tau)|
  \leq \varepsilon_0 \int_\tau^{\tau_T}
    \exp\!\left(\int_\tau^{\tau'} |s_x(\xi(\tau''),\tau'')|\,d\tau''\right) d\tau',
\end{equation}
where $\xi(\tau'')$ lies between $x(\tau'')$ and $\hat{x}(\tau'')$.
\end{theorem}

\begin{proof}
The trajectory error satisfies the differential equation (with $\tau$ decreasing):
\begin{equation}\label{eq:error-ode}
  \frac{de}{d\tau}
  = -s(x,\tau) + \hat{s}(\hat{x},\tau)
  = -\bigl[s(x,\tau) - s(\hat{x},\tau)\bigr] - \bigl[s(\hat{x},\tau) - \hat{s}(\hat{x},\tau)\bigr].
\end{equation}
By the mean value theorem, $s(x,\tau) - s(\hat{x},\tau) = s_x(\xi,\tau)\,(x - \hat{x})$ for some $\xi$ between $x$ and $\hat{x}$.
Thus $de/d\tau = -s_x(\xi,\tau)\,e + \epsilon(\hat{x},\tau)$, where $\epsilon = \hat{s} - s$ satisfies $|\epsilon| \leq \varepsilon_0$.

Switching to forward-in-reverse time $\sigma = \tau_T - \tau$:
\begin{equation}\label{eq:error-forward}
  \frac{de}{d\sigma} = s_x(\xi, \tau_T - \sigma)\,e + \epsilon(\hat{x}, \tau_T - \sigma).
\end{equation}
This is a linear inhomogeneous ODE.
The variation of constants formula~\citep{gronwall1919note} gives
\[
  e(\sigma)
  = \int_0^\sigma \epsilon(\sigma')\,
    \exp\!\left(\int_{\sigma'}^\sigma s_x(\xi(\sigma''),\tau_T - \sigma'')\,d\sigma''\right) d\sigma'.
\]
Bounding $|\epsilon| \leq \varepsilon_0$ and reverting to $\tau$-time yields~\eqref{eq:gronwall-bound}.
\end{proof}

\subsection{The amplification exponent in closed form}\label{sec:amp:exponent}

For the symmetric binary mixture, the integral $\int |s_x|\,d\tau$ in the bound~\eqref{eq:gronwall-bound} can then be evaluated exactly.

\begin{theorem}[Amplification exponent]\label{thm:amplification}
For the symmetric binary GMM~\eqref{eq:binary-gmm}, the amplification exponent for a trajectory through the shock center is
\begin{equation}\label{eq:Lambda}
  \Lambda(\tau)
  \coloneqq \int_\tau^{\tau^\ast} s_x(0,\tau')\,d\tau'
  = \frac{1}{2}\!\left[\frac{a^2}{\sigma_\tau^2} - 1 - \ln\!\frac{a^2}{\sigma_\tau^2}\right]
\end{equation}
for $\tau < \tau^\ast$, where $\sigma_\tau^2 = \sigma_0^2 + 2\tau$.
The amplification factor is $\exp(\Lambda)$.
\end{theorem}

\begin{proof}
From \Cref{prop:sx-origin}, $s_x(0,\tau') = (a^2 - \sigma_{\tau'}^2)/\sigma_{\tau'}^4$ with $\sigma_{\tau'}^2 = \sigma_0^2 + 2\tau'$.
Substituting $w = \sigma_0^2 + 2\tau'$ (so $dw = 2\,d\tau'$, and the limits transform as $\tau' = \tau \mapsto w = \sigma_\tau^2$ and $\tau' = \tau^\ast \mapsto w = a^2$):
\begin{align}
  \Lambda(\tau)
  &= \int_{\sigma_\tau^2}^{a^2} \frac{a^2 - w}{w^2}\,\frac{dw}{2}
   = \frac{1}{2}\int_{\sigma_\tau^2}^{a^2}\!\left(\frac{a^2}{w^2} - \frac{1}{w}\right) dw\notag\\
  &= \frac{1}{2}\!\left[-\frac{a^2}{w} - \ln w\right]_{\sigma_\tau^2}^{a^2}
   = \frac{1}{2}\!\left[\left(-1 - \ln a^2\right) - \left(-\frac{a^2}{\sigma_\tau^2} - \ln \sigma_\tau^2\right)\right]\notag\\
  &= \frac{1}{2}\!\left[\frac{a^2}{\sigma_\tau^2} - 1 - \ln\!\frac{a^2}{\sigma_\tau^2}\right]. \qedhere
\end{align}
\end{proof}

\begin{corollary}[Asymptotic amplification]\label{cor:asymptotic}
Define the signal-to-noise ratio $\mathrm{SNR} = a^2/\sigma_\tau^2$.
Then:
\begin{equation}\label{eq:Lambda-asymptotic}
  \Lambda \approx \frac{\mathrm{SNR}}{2}
  \qquad \text{for } \mathrm{SNR} \gg 1.
\end{equation}
The amplification factor grows as $\exp(a^2/(2\sigma_\tau^2))$, which is exponential in the SNR.
\end{corollary}

\begin{figure}[t]
\centering
\includegraphics[width=\linewidth]{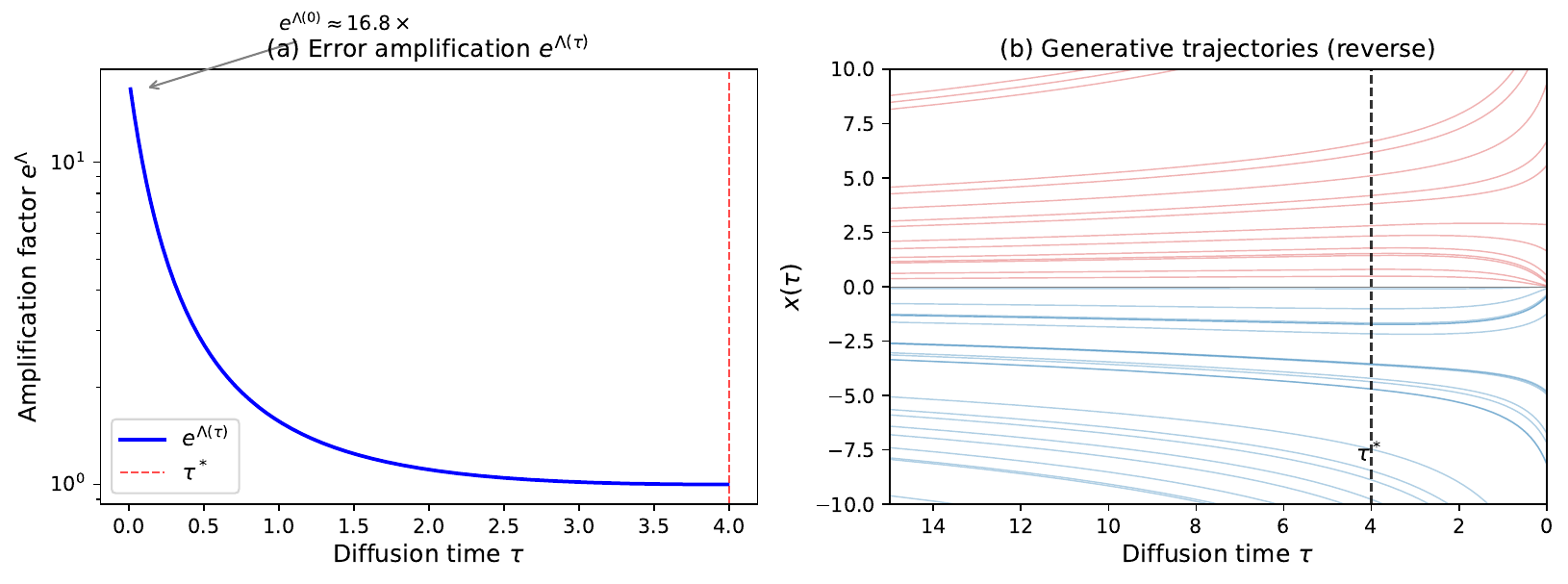}
\caption{Amplification near the transition. Panel~(a) plots the factor $e^{\Lambda(\tau)}$ on a log scale; at $\tau=0$ it is about $18$. Panel~(b) traces reverse-time probability-flow trajectories. Above $\tau^\ast$ they stay mixed, while below $\tau^\ast$ they split toward the two basins.}
\label{fig:amplification-trajectories}
\end{figure}

\begin{proof}
For $\mathrm{SNR} \gg 1$: $\ln(\mathrm{SNR}) \ll \mathrm{SNR}$ and the constant $-1$ is negligible, giving $\Lambda \approx \mathrm{SNR}/2$.
\end{proof}

\begin{remark}[Numerical illustration]\label{rem:numerical-amp}
For $a = 3$, $\sigma_0 = 1$, at $\tau = 0$: $\mathrm{SNR} = 9$, $\Lambda(0) = \tfrac{1}{2}(9 - 1 - \ln 9) \approx 2.90$, and $\exp(\Lambda) \approx 18.2$.
Score errors near the mode boundary are amplified by a factor of approximately $18$ relative to errors in the smooth (single-mode) region.
This amplification is captured by the Burgers interfacial analysis above and quantifies the well-known empirical observation~\citep{song2020improved, karras2022elucidating} that diffusion models are sensitive to score accuracy at low noise levels.
\end{remark}

\subsection{KL and total variation bounds}\label{sec:amp:kl}

We connect the trajectory-level amplification to distributional error bounds for the reverse-time SDE.

\begin{proposition}[KL bound for the reverse-time SDE; cf.\ {\citealp{chen2023sampling}}]\label{prop:kl-bound}
Let $\hat{p}_0^{\mathrm{SDE}}$ denote the distribution generated by the reverse-time SDE~\eqref{eq:reverse-sde} when the true score $s$ is replaced by an approximate score $\hat{s}$.
Then
\begin{equation}\label{eq:kl-bound}
  \KL(\hat{p}_0^{\mathrm{SDE}} \| p_0)
  \leq \frac{1}{2}\int_0^{\tau_T} \E_{p_\tau}\!\bigl[\|\hat{s}(\bm{x},\tau) - s(\bm{x},\tau)\|^2\bigr]\,d\tau.
\end{equation}
\end{proposition}

This follows from the Girsanov theorem~\citep{girsanov1960transforming} applied to the reverse-time SDE~\eqref{eq:reverse-sde}; see \citet[Theorem~1]{chen2023sampling} for the rigorous statement.
By Pinsker's inequality~\citep{tsybakov2009introduction}:
\begin{equation}\label{eq:pinsker}
  \TV(\hat{p}_0^{\mathrm{SDE}}, p_0)
  \leq \sqrt{\tfrac{1}{2}\,\KL(\hat{p}_0^{\mathrm{SDE}} \| p_0)}
  \leq \frac{1}{2}\!\left(\int_0^{\tau_T} \E_{p_\tau}\!\bigl[\|\hat{s} - s\|^2\bigr]\,d\tau\right)^{\!1/2}.
\end{equation}

\begin{definition}[Interfacial and regular regions]\label{def:regions}
For a $K$-component GMM, define the \emph{interfacial region} at time $\tau$ as the set of points within one interfacial width of any inter-mode boundary:
\begin{equation}\label{eq:shock-region}
  \mathcal{S}_\delta(\tau) = \bigl\{x \in \R : \min_j |x - x_j^\ast(\tau)| < \delta(\tau)\bigr\},
\end{equation}
where $x_j^\ast(\tau)$ are the boundary locations and $\delta(\tau) = \sigma_\tau^2/a$ is the interfacial width~\eqref{eq:shock-width}.
The \emph{regular region} is $\mathcal{R}(\tau) = \R \setminus \mathcal{S}_\delta(\tau)$.
\end{definition}

\begin{proposition}[Score regularity by region]\label{prop:regularity}
In the regular region, the score is smooth with $\|s_x\|_{L^\infty(\mathcal{R}(\tau))} = O(\sigma_\tau^{-2})$.
In the interfacial region, $\|s_x\|_{L^\infty(\mathcal{S}_\delta(\tau))} = O(a^2/\sigma_\tau^4)$.
\end{proposition}

\begin{proof}
In $\mathcal{R}(\tau)$, the density is dominated by a single Gaussian component, so $s(x,\tau) \approx -(x - \mu_k)/\sigma_\tau^2$ and $s_x \approx -1/\sigma_\tau^2$.
In $\mathcal{S}_\delta(\tau)$, by \Cref{prop:sx-origin}, $|s_x(0,\tau)| = |a^2 - \sigma_\tau^2|/\sigma_\tau^4 \sim a^2/\sigma_\tau^4$ for $\tau \ll \tau^\ast$.
\end{proof}

The practical implication is that the interfacial region is spatially narrow (width $O(\sigma_\tau^2/a)$) yet contains the steepest score gradients (of order $a^2/\sigma_\tau^4$ rather than the $\sigma_\tau^{-2}$ of the regular region).
In the present analysis, this $a^2/\sigma_\tau^2$-fold ratio is the key quantity driving the Gr\"onwall exponent~\eqref{eq:gronwall-bound}, and hence one concrete mechanism by which mode-boundary score errors degrade sample quality.

\section{Multi-Dimensional Extension}\label{sec:multidim}

Having isolated the exact boundary-normal mechanism in \Cref{sec:shocks}, we now separate two complementary higher-dimensional questions. The first is intrinsic and distribution-free: the full vector Burgers dynamics and its curl-free structure in $\R^d$. The second is model-specific: how the local criterion specializes in Gaussian mixtures to explicit geometric objects such as Voronoi boundaries and leading-order spectral thresholds.

\subsection{The vector Burgers system}\label{sec:md:vector}

\begin{theorem}[Score PDE in $\R^d$]\label{thm:score-pde-d}
Let $p(\bm{x},\tau)$ be a positive smooth solution of the heat equation $\partial_\tau p = \Delta p$ in $\R^d$.
Then each component $s_i(\bm{x},\tau) = \partial_i \log p(\bm{x},\tau)$ of the score satisfies
\begin{equation}\label{eq:score-pde-d}
  \frac{\partial s_i}{\partial \tau}
  = \Delta s_i + 2\,s_k\,\partial_k s_i
  \qquad (i = 1,\ldots,d),
\end{equation}
where Einstein summation over $k$ is implied.
In vector notation:
\begin{equation}\label{eq:vector-burgers-score}
  \partial_\tau \bm{s} = \Delta \bm{s} + 2\,(\bm{s} \cdot \nabla)\bm{s}.
\end{equation}
Under $\bm{u} = -2\bm{s}$, this becomes the $d$-dimensional viscous Burgers system:
\begin{equation}\label{eq:vector-burgers}
  \partial_\tau \bm{u} + (\bm{u} \cdot \nabla)\bm{u} = \Delta \bm{u}.
\end{equation}
\end{theorem}

\begin{proof}
The one-dimensional argument of \Cref{thm:score-pde} extends component-wise.
We use the identities $\partial_i p = s_i\,p$ and $\partial_i \partial_j p = (\partial_j s_i + s_i s_j)\,p$ (by direct computation, as in~\eqref{eq:pxx}).
The Laplacian of $p$ is $\Delta p = (\partial_k s_k + |\bm{s}|^2)\,p$.
Applying $\partial_i$ to $\Delta p$:
\begin{equation}\label{eq:partial-i-laplacian}
  \partial_i(\Delta p)
  = \bigl(\partial_i \partial_k s_k + 2\,s_m\,\partial_i s_m + s_i\,\partial_k s_k + s_i\,|\bm{s}|^2\bigr)\,p.
\end{equation}
From $\partial_\tau s_i = (\partial_i \Delta p)/p - s_i\,(\Delta p)/p$ (the $d$-dimensional analogue of~\eqref{eq:s-tau-quotient}):
\begin{align}
  \partial_\tau s_i
  &= \partial_i\partial_k s_k + 2\,s_m\,\partial_i s_m + s_i\,\partial_k s_k + s_i|\bm{s}|^2 \\
  &\qquad - s_i\bigl(\partial_k s_k + |\bm{s}|^2\bigr)\notag\\
  &= \partial_i\partial_k s_k + 2\,s_m\,\partial_i s_m.\label{eq:general-score-pde}
\end{align}
Since $s_i = \partial_i \log p$, we have $\partial_k s_i = \partial_i s_k$ (symmetry of mixed partials), hence $\partial_i \partial_k s_k = \partial_k \partial_k s_i = \Delta s_i$.
Therefore~\eqref{eq:general-score-pde} reduces to~\eqref{eq:score-pde-d}.

The Burgers form~\eqref{eq:vector-burgers} follows from $\bm{u} = -2\bm{s}$ by the same algebra as \Cref{thm:score-burgers}.
\end{proof}

\begin{remark}\label{rem:vector-burgers}
The system~\eqref{eq:vector-burgers} is precisely the $d$-dimensional viscous Burgers equation studied in fluid dynamics as a model for irrotational compressible flow~\citep{whitham1974linear}.
The multi-dimensional Cole--Hopf transform $\bm{u} = -2\nabla\log\varphi$ with $\varphi_\tau = \Delta\varphi$ yields~\eqref{eq:vector-burgers}, confirming the identification.
\end{remark}

\subsection{Curl preservation}\label{sec:md:curl}

The true score is curl-free by construction ($\bm{s} = \nabla\log p$ implies $\partial_i s_j = \partial_j s_i$).
The next result shows that this property is preserved by the vector Burgers dynamics, even when the equation is viewed for a general vector field.

\begin{definition}[Vorticity]\label{def:vorticity}
For a vector field $\bm{v}$ on $\R^d$, the \emph{vorticity} is the antisymmetric tensor
\begin{equation}\label{eq:vorticity-def}
  \Omega_{ij} = \partial_i v_j - \partial_j v_i.
\end{equation}
The field $\bm{v}$ is irrotational (curl-free) if and only if $\Omega_{ij} = 0$ for all $i,j$.
In $d = 3$, the dual vector $(\nabla \times \bm{v})_i = \epsilon_{ijk}\Omega_{jk}/2$ is the usual curl~\citep{bhatia2013helmholtz}.
\end{definition}

\begin{theorem}[Vorticity equation for vector Burgers]\label{thm:vorticity}
If $\bm{v}$ satisfies the vector Burgers system $\partial_\tau v_i = \Delta v_i + 2\,v_k\,\partial_k v_i$, then the vorticity $\Omega_{ij}$ satisfies the linear parabolic system
\begin{equation}\label{eq:vorticity-eqn}
  \partial_\tau \Omega_{ij}
  = \Delta \Omega_{ij}
    + 2\,v_k\,\partial_k \Omega_{ij}
    + 2\,(\partial_i v_k)\,\Omega_{kj}
    - 2\,(\partial_j v_k)\,\Omega_{ki}.
\end{equation}
\end{theorem}

\begin{proof}
Apply $\partial_i$ to the Burgers equation for component $j$:
\begin{equation}\label{eq:diff-j}
  \partial_\tau(\partial_i v_j) = \Delta(\partial_i v_j) + 2\,(\partial_i v_k)(\partial_k v_j) + 2\,v_k\,\partial_k(\partial_i v_j).
\end{equation}
Interchange $i \leftrightarrow j$ and subtract:
\begin{equation}\label{eq:diff-ij}
  \partial_\tau \Omega_{ij}
  = \Delta \Omega_{ij}
    + 2\,v_k\,\partial_k \Omega_{ij}
    + 2\bigl[(\partial_i v_k)(\partial_k v_j) - (\partial_j v_k)(\partial_k v_i)\bigr].
\end{equation}
For the bracketed term, decompose $\partial_k v_j = \partial_j v_k + \Omega_{kj}$:
\[
  (\partial_i v_k)(\partial_k v_j)
  = (\partial_i v_k)(\partial_j v_k) + (\partial_i v_k)\,\Omega_{kj}.
\]
Similarly, $\partial_k v_i = \partial_i v_k + \Omega_{ki}$ gives
\[
  (\partial_j v_k)(\partial_k v_i)
  = (\partial_j v_k)(\partial_i v_k) + (\partial_j v_k)\,\Omega_{ki}.
\]
The symmetric terms $(\partial_i v_k)(\partial_j v_k)$ cancel upon subtraction, leaving~\eqref{eq:vorticity-eqn}.
\end{proof}

\begin{theorem}[Curl preservation]\label{thm:curl-preservation}
Let $\bm{v}$ be a smooth solution of the vector Burgers equation~\eqref{eq:vector-burgers-score} on $\R^d \times [0,T]$ with $\nabla \bm{v}$ bounded.
If $\Omega_{ij}(\bm{x},0) = 0$ for all $\bm{x} \in \R^d$ and all $i,j$, then $\Omega_{ij}(\bm{x},\tau) = 0$ for all $\tau \in [0,T]$.
\end{theorem}

\begin{proof}
Equation~\eqref{eq:vorticity-eqn} is a linear parabolic system in the unknowns $\{\Omega_{ij}\}$:
\begin{equation}\label{eq:linear-parabolic}
  \partial_\tau \Omega_{ij} = \Delta \Omega_{ij} + B_k(\bm{x},\tau)\,\partial_k \Omega_{ij} + C_{ij,mn}(\bm{x},\tau)\,\Omega_{mn},
\end{equation}
where $B_k = 2v_k$ and $C$ collects the zero-order terms from~\eqref{eq:vorticity-eqn}.
Both $B$ and $C$ are bounded on $[0,T]$ by assumption.

Define the energy $E(\tau) = \tfrac{1}{2}\int_{\R^d} |\Omega|^2\,d\bm{x} = \tfrac{1}{2}\int \Omega_{ij}\Omega_{ij}\,d\bm{x}$.
Differentiating under the integral and substituting~\eqref{eq:vorticity-eqn}:
\begin{align}
  \frac{dE}{d\tau}
  &= \int \Omega_{ij}\,\partial_\tau\Omega_{ij}\,d\bm{x}\notag\\
  &= \underbrace{\int \Omega_{ij}\,\Delta\Omega_{ij}\,d\bm{x}}_{= -\int |\nabla\Omega|^2\,d\bm{x} \leq 0}
    + 2\underbrace{\int \Omega_{ij}\,v_k\,\partial_k\Omega_{ij}\,d\bm{x}}_{= -\int (\partial_k v_k)|\Omega|^2\,d\bm{x}\,/\,2}
    + \text{zero-order terms}.\label{eq:energy-est}
\end{align}
The first integral is non-positive (integration by parts with vanishing boundary terms).
The second follows from integration by parts: $\int \Omega_{ij} v_k \partial_k \Omega_{ij} = -\frac{1}{2}\int (\partial_k v_k)|\Omega|^2$.
The zero-order terms satisfy $|\int \Omega_{ij}(\partial_i v_k)\Omega_{kj}| \leq \|\nabla \bm{v}\|_\infty \int |\Omega|^2 = 2\|\nabla \bm{v}\|_\infty E$, and similarly for the $(\partial_j v_k)\Omega_{ki}$ term.
Combining:
\begin{equation}\label{eq:gronwall-energy}
  \frac{dE}{d\tau} \leq M(\tau)\,E(\tau),
  \qquad M(\tau) = \|\nabla \cdot \bm{v}\|_\infty + 4\|\nabla \bm{v}\|_\infty.
\end{equation}
By the Gr\"onwall inequality~\citep{gronwall1919note}:
\[
  E(\tau) \leq E(0)\,\exp\!\left(\int_0^\tau M(\tau')\,d\tau'\right).
\]
Since $E(0) = 0$, we conclude $E(\tau) = 0$ for all $\tau \in [0,T]$.
As $|\Omega|^2 \geq 0$ with zero integral, $\Omega_{ij}(\bm{x},\tau) = 0$ almost everywhere, and by continuity (smoothness of the solution for $\tau > 0$, guaranteed by the heat kernel~\citep{evans2010partial}), everywhere.
\end{proof}

\begin{corollary}[Non-conservative scores are approximation artifacts]\label{cor:curl-artifact}
The true score $\bm{s} = \nabla\log p$ of a diffusion model is curl-free for all $\tau > 0$, and the vector Burgers dynamics~\eqref{eq:vector-burgers-score} preserves this property.
Any non-zero vorticity $\Omega_{ij}$ measured in a learned score network $\bm{s}_\theta$~\citep{vuong2025score, lai2023fp} is entirely attributable to the neural network approximation error.
\end{corollary}

This geometry is illustrated in \Cref{fig:2d-score-curl}: the two-dimensional score field has a sharp directional transition across the inter-mode boundary, yet its measured curl remains numerically zero.

\subsection{Shock surfaces in \texorpdfstring{$\R^d$}{Rd}}\label{sec:md:shocks}

In $d > 1$, the formal inviscid or low-noise Burgers description leads to \emph{shock surfaces}---codimension-$1$ manifolds across which the score becomes discontinuous in the limiting picture.

\begin{proposition}[Shock surfaces as Voronoi boundaries]\label{prop:voronoi}
For the equal-covariance GMM~\eqref{eq:gmm-def} with equal weights $w_k = 1/K$ in the limit $\sigma_\tau \to 0$, the limiting shock surfaces of the score are given by the faces of the Voronoi tessellation generated by the means $\{\bm{\mu}_k\}$:
\begin{equation}\label{eq:voronoi}
  \Gamma_{jk}
  = \bigl\{\bm{x} \in \R^d : |\bm{x} - \bm{\mu}_j| = |\bm{x} - \bm{\mu}_k|\bigr\}
  = \bigl\{\bm{x} : (\bm{\mu}_j - \bm{\mu}_k) \cdot \bm{x} = \tfrac{|\bm{\mu}_j|^2 - |\bm{\mu}_k|^2}{2}\bigr\}.
\end{equation}
For unequal weights, the boundaries are the \emph{weighted Voronoi} (power diagram) faces.
\end{proposition}

\begin{figure}[t]
\centering
\includegraphics[width=\linewidth]{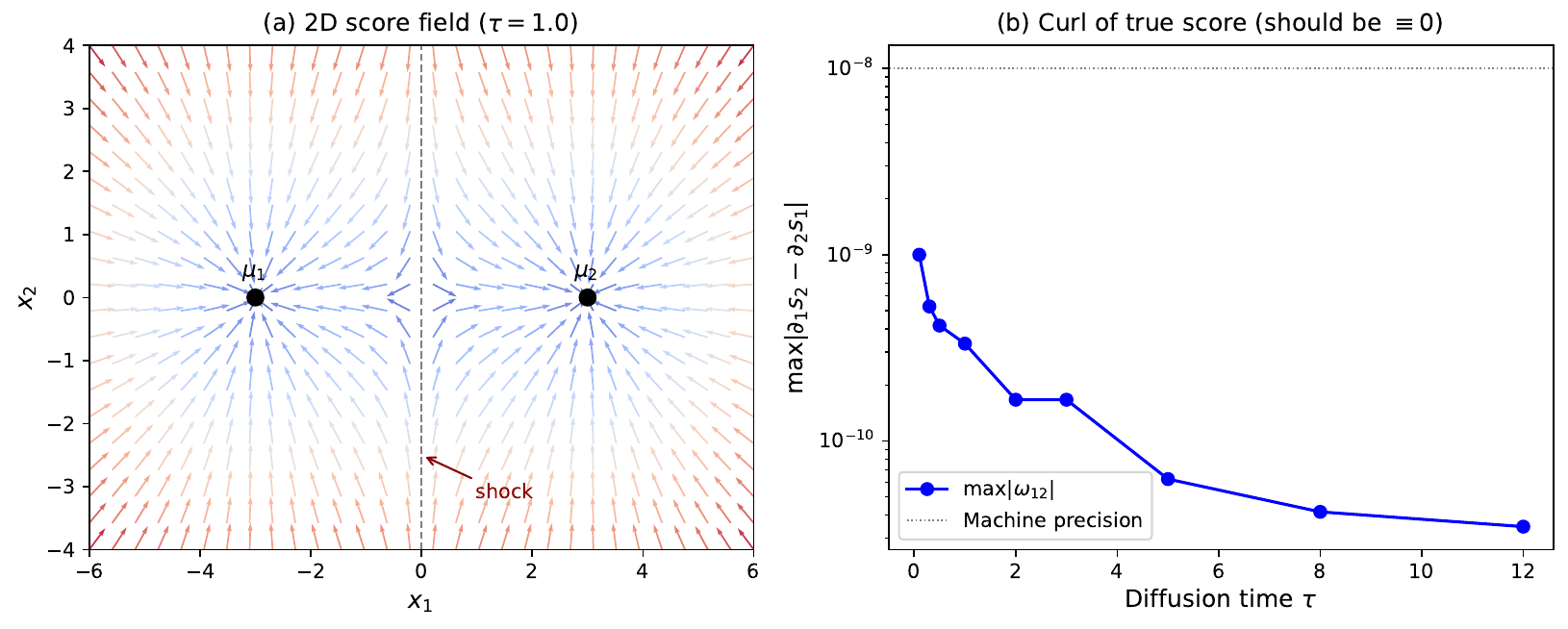}
\caption{Two-dimensional score geometry. Panel~(a) depicts the score field for a two-component Gaussian mixture in $\R^2$ at $\tau=1$; the dashed line marks the inter-mode boundary. Panel~(b) tracks the maximum curl magnitude $|\partial_1 s_2 - \partial_2 s_1|$ across random test points at several diffusion times. The values stay below $10^{-9}$ throughout.}
\label{fig:2d-score-curl}
\end{figure}

\begin{proof}
As $\sigma_\tau \to 0$, the posterior responsibility $r_k(\bm{x},\tau) \to \mathbf{1}[k = \argmax_m w_m\,\mathcal{N}(\bm{x};\bm{\mu}_m,\sigma_\tau^2\bm{I})]$, which for equal weights reduces to $k = \argmin_m |\bm{x} - \bm{\mu}_m|$.
On each Voronoi cell, $s(\bm{x},\tau) \approx -(\bm{x} - \bm{\mu}_k)/\sigma_\tau^2$---a smooth field pointing toward the nearest mean.
Across a Voronoi face $\Gamma_{jk}$, the score jumps discontinuously from the $\bm{\mu}_j$-directed field to the $\bm{\mu}_k$-directed field.
In this low-noise inviscid description, these discontinuities are the relevant shock surfaces of the vector Burgers equation.
\end{proof}

\subsection{A Gaussian-mixture specialization of the local criterion in \texorpdfstring{$\R^d$}{Rd}}\label{sec:md:speciation}

\begin{proposition}[Leading-order Gaussian-mixture specialization in $\R^d$]\label{thm:d-speciation}
For the equal-covariance GMM~\eqref{eq:gmm-def}, the exact local criterion of \Cref{thm:local-speciation} can be expanded explicitly at the weighted mean $\bar{\bm{x}} = \sum_k w_k \bm{\mu}_k$. In the high-noise limit $\sigma_\tau^2 \gg \lambda_1(\bm{W})$, the score Jacobian $\bm{J}(\bm{x},\tau) = \nabla \bm{s}(\bm{x},\tau)$ satisfies:
\begin{equation}\label{eq:jacobian-high-noise}
  \bm{J}(\bar{\bm{x}},\tau) \approx -\frac{\bm{I}}{\sigma_\tau^2} + \frac{\bm{W}}{\sigma_\tau^4} + O(\sigma_\tau^{-6}),
\end{equation}
where $\bm{W}$ is the between-class covariance~\eqref{eq:W-def}.
The eigenvalues of $\bm{J}$ are
\begin{equation}\label{eq:jacobian-eigs}
  \lambda_i^{(J)} = \frac{\lambda_i^{(W)} - \sigma_\tau^2}{\sigma_\tau^4} + O(\sigma_\tau^{-6}).
\end{equation}
The first speciation is predicted at leading order along the leading eigenvector $\bm{e}_1$ of $\bm{W}$ when $\lambda_1^{(J)} \approx 0$, at the critical time
\begin{equation}\label{eq:tau-star-d}
  \tau^\ast_{\mathrm{LO}} = \frac{\lambda_1(\bm{W}) - \sigma_0^2}{2}.
\end{equation}
This leading-order threshold coincides with the spectral criterion of \citet{biroli2024dynamical} and becomes exact when the posterior responsibilities remain equal at $\bar{\bm{x}}$ (see \Cref{sec:correction}).
For hierarchical data with $\lambda_1 > \lambda_2 > \cdots$, the leading-order cascade is $\tau_{i,\mathrm{LO}}^\ast = (\lambda_i(\bm{W}) - \sigma_0^2)/2$, matching the hierarchical phase transitions of \citet{sclocchi2024phase} at this order.
\end{proposition}

\begin{proof}
The score of the GMM~\eqref{eq:gmm-noised} at $\bm{x}$ can be written as
$\bm{s}(\bm{x},\tau) = -\bm{x}/\sigma_\tau^2 + \sigma_\tau^{-2}\sum_k r_k(\bm{x},\tau)\,\bm{\mu}_k$,
where $r_k = w_k\mathcal{N}(\bm{x};\bm{\mu}_k,\sigma_\tau^2\bm{I})/\sum_m w_m\mathcal{N}(\bm{x};\bm{\mu}_m,\sigma_\tau^2\bm{I})$ are the posterior responsibilities.
Differentiating $r_k$ with respect to $x_j$ and evaluating at $\bar{\bm{x}}$:
\[
  \partial_j r_k\big|_{\bar{\bm{x}}} = \frac{r_k}{\sigma_\tau^2}\bigl[\mu_{k,j} - \tilde{\mu}_j\bigr],
\]
where $\tilde{\bm{\mu}} = \sum_m r_m \bm{\mu}_m$ is the posterior mean.
The Jacobian is then $J_{ij} = -\delta_{ij}/\sigma_\tau^2 + C_{ij}/\sigma_\tau^4$, where $C_{ij} = \sum_k r_k \mu_{k,i}\mu_{k,j} - \tilde{\mu}_i\tilde{\mu}_j$ is the posterior covariance of the means.

In the high-noise limit, $r_k \to w_k$, $\tilde{\bm{\mu}} \to \bar{\bm{x}}$, and $C_{ij} \to W_{ij}$, giving~\eqref{eq:jacobian-high-noise}.
The eigenvalues follow immediately.
Setting the leading-order approximation $\lambda_1^{(J)} \approx 0$ gives $\sigma_\tau^2 = \lambda_1(\bm{W})$, i.e., $\tau^\ast_{\mathrm{LO}} = (\lambda_1(\bm{W}) - \sigma_0^2)/2$.

The connection to \citet{biroli2024dynamical} follows because their speciation criterion is $\lambda_1(\bm{W})/\sigma_\tau^2 = 1$~\citep[Eq.~(7)]{biroli2024dynamical}, which is equivalent at this order.
For hierarchical speciation, each eigenvalue $\lambda_i$ crossing $\sigma_\tau^2$ triggers a new leading-order bifurcation along $\bm{e}_i$, matching the cascade described by \citet{sclocchi2024phase}.
\end{proof}

\begin{remark}[Matrix Riccati structure]\label{rem:riccati}
Along the inviscid vector Burgers characteristics through $\bar{\bm{x}}$ (where $\bm{s} \approx \bm{0}$), the Jacobian $\bm{J}$ satisfies the matrix Riccati equation $d\bm{J}/d\sigma = -2\bm{J}^2$ (with $\sigma = \tau_T - \tau$).
For a symmetric matrix $\bm{J}$ with eigenvalues $\lambda_i(0) < 0$ (unimodal regime), each eigenvalue evolves as $\lambda_i(\sigma) = 1/(2\sigma + 1/\lambda_i(0))$, which diverges at $\sigma_i^\ast = -1/(2\lambda_i(0))$.
The first divergence determines the corresponding leading-order threshold, yielding the same $\tau^\ast_{\mathrm{LO}}$ as above.
\end{remark}

\section{The VP-SDE via Coordinate Reduction}\label{sec:vp}

The VP-SDE~\eqref{eq:vp-sde} introduces a mean-reverting drift $-\tfrac{1}{2}\beta(t)\bm{x}$ in addition to diffusion, leading to a \emph{forced} Burgers equation for the score.
An exact coordinate transformation reduces the VP analysis to the VE case studied in the preceding sections, yielding closed-form speciation times and interfacial widths.

\subsection{The VP score PDE}\label{sec:vp:pde}

For reference, we record the VP score PDE in one dimension.

\begin{theorem}[VP score PDE]\label{thm:vp-score-pde}
Under the VP forward process~\eqref{eq:vp-sde} in $d=1$, the score $s(x,t) = \partial_x \log p(x,t)$ satisfies
\begin{equation}\label{eq:vp-score-pde}
  \frac{\partial s}{\partial t}
  = \frac{\beta(t)}{2}\!\left[
      \frac{\partial^2 s}{\partial x^2}
    + 2\,s\,\frac{\partial s}{\partial x}
    + x\,\frac{\partial s}{\partial x}
    + s
  \right].
\end{equation}
\end{theorem}

\begin{proof}
From the VP Fokker--Planck equation~\eqref{eq:vp-fpe} with $\nu = \beta(t)/2$:
$\partial_t p = \nu[p + x\,\partial_x p + \partial_x^2 p] = \nu[1 + xs + s_x + s^2]\,p,$
where we used $\partial_x p = sp$ and $\partial_x^2 p = (s_x + s^2)p$.
Define $A = 1 + xs + s_x + s^2$.
Then $\partial_t(\partial_x p) = \nu\,\partial_x(Ap) = \nu(A_x + As)\,p$, where
$A_x = s + xs_x + s_{xx} + 2ss_x$.
Hence:
\begin{align*}
  \partial_t s
  &= \frac{\partial_t(\partial_x p)}{p} - s\,\frac{\partial_t p}{p}
   = \nu\bigl[A_x + As\bigr] - \nu\,s\,A\\
  &= \nu\,A_x
   = \nu\bigl[s + xs_x + s_{xx} + 2ss_x\bigr]
   = \frac{\beta(t)}{2}\bigl[s_{xx} + 2ss_x + xs_x + s\bigr]. \qedhere
\end{align*}
\end{proof}

\begin{remark}[Structure]\label{rem:vp-structure}
Equation~\eqref{eq:vp-score-pde} decomposes as
\begin{equation}\label{eq:vp-decomposition}
  \partial_t s
  = \underbrace{\frac{\beta}{2}\bigl(s_{xx} + 2\,s\,s_x\bigr)}_{\text{Burgers (VE)}}
  + \underbrace{\frac{\beta}{2}\bigl(x\,s_x + s\bigr)}_{\text{OU forcing}}
  = \frac{\beta}{2}\,\frac{\partial}{\partial x}\bigl[s_x + s^2 + x\,s\bigr],
\end{equation}
where the OU forcing $xs_x + s = \partial_x(xs)$ acts as a source term.
The Cole--Hopf variable $u = -2s$ satisfies a forced Burgers equation with linear advection and growth~\citep{whitham1974linear}.
Rather than analyzing this forced equation directly, we reduce it to the pure VE case via a coordinate transformation.
\end{remark}

\subsection{The rescaling transformation}\label{sec:vp:transform}

\begin{definition}[Effective diffusion time]\label{def:tau-eff}
For the VP-SDE, recall the signal attenuation $\alpha(t)$ from \Cref{sec:prelim}. Define the \emph{effective VE diffusion time}:
\begin{equation}\label{eq:tau-eff}
  \tau_{\mathrm{eff}}(t)
  = \frac{1 - \alpha(t)^2}{2\,\alpha(t)^2}.
\end{equation}
\end{definition}

\begin{lemma}[Density under rescaling]\label{lem:rescaling}
Define the rescaled variable $Z_t = X_t / \alpha(t)$.
Then the density $q_t(z)$ of $Z_t$ satisfies $q_t = p_0 * G_{\tau_{\mathrm{eff}}(t)}$, i.e., $q_t$ solves the VE heat equation at effective time $\tau_{\mathrm{eff}}(t)$.
\end{lemma}

\begin{proof}
The VP conditional is $X_t \mid X_0 \sim \mathcal{N}(\alpha_t X_0,\, (1-\alpha_t^2)\bm{I})$~\citep{song2021score}.
Thus $Z_t \mid X_0 \sim \mathcal{N}(X_0,\, (1-\alpha_t^2)/\alpha_t^2\,\bm{I})$.
The marginal density of $Z_t$ is
$q_t(z) = \int p_0(y)\,\mathcal{N}(z;\,y,\,(1-\alpha_t^2)/\alpha_t^2\,\bm{I})\,dy = (p_0 * G_{\tau_{\mathrm{eff}}})(z)$,
where the Gaussian kernel has variance $(1-\alpha_t^2)/\alpha_t^2 = 2\tau_{\mathrm{eff}}(t)$.
\end{proof}

\begin{theorem}[VP--VE score equivalence]\label{thm:vp-ve}
The VP and VE scores are related by
\begin{equation}\label{eq:vp-ve-score}
  s_{\mathrm{VP}}(x,t) = \frac{1}{\alpha(t)}\,s_{\mathrm{VE}}\!\left(\frac{x}{\alpha(t)},\, \tau_{\mathrm{eff}}(t)\right),
\end{equation}
where $s_{\mathrm{VE}}(z,\tau) = \partial_z \log(p_0 * G_\tau)(z)$ is the VE score satisfying the pure Burgers equation~\eqref{eq:score-pde}.
\end{theorem}

\begin{proof}
By the change-of-variables formula, $p_t(x) = \alpha_t^{-1}\,q_t(x/\alpha_t)$ (in $d = 1$; in $d$ dimensions, $\alpha_t^{-d}$).
Therefore:
\[
  s_{\mathrm{VP}}(x,t)
  = \partial_x \log p_t(x)
  = \partial_x \log q_t(x/\alpha_t)
  = \frac{1}{\alpha_t}\,(\partial_z \log q_t)\big|_{z=x/\alpha_t}
  = \frac{1}{\alpha_t}\,s_{\mathrm{VE}}(x/\alpha_t, \tau_{\mathrm{eff}}(t)),
\]
where we used \Cref{lem:rescaling} to identify $q_t$ with the VE density at time $\tau_{\mathrm{eff}}$.
\end{proof}

\subsection{VP speciation time}\label{sec:vp:speciation}

\begin{corollary}[VP speciation time]\label{cor:vp-speciation}
For the symmetric binary GMM~\eqref{eq:binary-gmm} under the VP-SDE with constant $\beta$, the speciation time satisfies
\begin{equation}\label{eq:vp-speciation}
  \tau_{\mathrm{eff}}(t^\ast_{\mathrm{VP}}) = \tau^\ast_{\mathrm{VE}} = \frac{a^2 - \sigma_0^2}{2}.
\end{equation}
Solving for $t^\ast_{\mathrm{VP}}$:
\begin{equation}\label{eq:vp-speciation-explicit}
  t^\ast_{\mathrm{VP}}
  = \frac{1}{\beta}\,\ln\!\bigl(a^2 - \sigma_0^2 + 1\bigr).
\end{equation}
\end{corollary}

\begin{proof}
From \Cref{thm:vp-ve}, the VP speciation occurs when the VE score (in $z$-coordinates) reaches the same speciation threshold, i.e., at VE diffusion time $\tau^\ast_{\mathrm{VE}}$.
Setting $\tau_{\mathrm{eff}}(t) = \tau^\ast_{\mathrm{VE}}$:
\[
  \frac{1 - \alpha^2}{2\alpha^2} = \frac{a^2 - \sigma_0^2}{2}
  \;\;\Longrightarrow\;\;
  1 - \alpha^2 = \alpha^2(a^2 - \sigma_0^2)
  \;\;\Longrightarrow\;\;
  \alpha^2 = \frac{1}{a^2 - \sigma_0^2 + 1}.
\]
For constant $\beta$: $\alpha(t) = e^{-\beta t/2}$, so $e^{-\beta t} = 1/(a^2 - \sigma_0^2 + 1)$, giving~\eqref{eq:vp-speciation-explicit}.
\end{proof}

\subsection{VP interfacial width}\label{sec:vp:profile}

\begin{corollary}[VP interfacial width]\label{cor:vp-shock}
The background-subtracted VP score layer at $x = 0$ for the symmetric binary mixture has width (in $x$-space):
\begin{equation}\label{eq:vp-shock-width}
  \delta_{\mathrm{VP}}(t)
  = \alpha(t) \cdot \frac{\sigma_{\tau_{\mathrm{eff}}}^2}{a}
  = \frac{1 - \alpha(t)^2(1 - \sigma_0^2)}{a\,\alpha(t)}.
\end{equation}
\end{corollary}

\begin{proof}
By \Cref{thm:vp-ve}, the VP score at $x$ is the VE score at $z = x/\alpha$ rescaled by $1/\alpha$.
The VE interfacial layer has width $\delta_{\mathrm{VE}} = \sigma_{\tau_{\mathrm{eff}}}^2/a$ in $z$-space (\Cref{prop:shock-profile}).
Mapping back to $x$-space: $\delta_{\mathrm{VP}} = \alpha\,\delta_{\mathrm{VE}}$.
Computing $\sigma_{\tau_{\mathrm{eff}}}^2 = \sigma_0^2 + 2\tau_{\mathrm{eff}} = \sigma_0^2 + (1-\alpha^2)/\alpha^2 = (1 - \alpha^2(1-\sigma_0^2))/\alpha^2$ gives~\eqref{eq:vp-shock-width}.
\end{proof}

\subsection{Summary: VP reduces to VE}\label{sec:vp:summary}

The key message of this section is that, for the results studied here, no separate analysis of the forced Burgers equation~\eqref{eq:vp-score-pde} is needed.
The rescaling $Z = X/\alpha(t)$ absorbs the OU drift entirely, reducing the VP score to a rescaled VE score.
Under this transformation, the VE Burgers correspondence (\Cref{thm:score-burgers}), the background-subtracted interfacial profile (\Cref{prop:shock-profile}), the speciation criterion (\Cref{thm:shock-speciation}), the error amplification (\Cref{thm:amplification}), and the curl preservation (\Cref{thm:curl-preservation}) translate directly to the VP setting.
\Cref{fig:vp-ve} makes this equivalence concrete: the transformed and direct VP scores overlap to machine precision, and the effective-time map sends the VP critical time exactly to the VE speciation time.

This unification has a practical consequence: noise schedule optimization for VP models~\citep{kingma2021variational, karras2022elucidating} can be analyzed entirely in the VE Burgers framework by working in the effective time~\eqref{eq:tau-eff}, reducing the design problem to choosing $\tau_{\mathrm{eff}}(t)$ to optimally traverse the interfacial layer.

\begin{figure}[t]
\centering
\includegraphics[width=\linewidth]{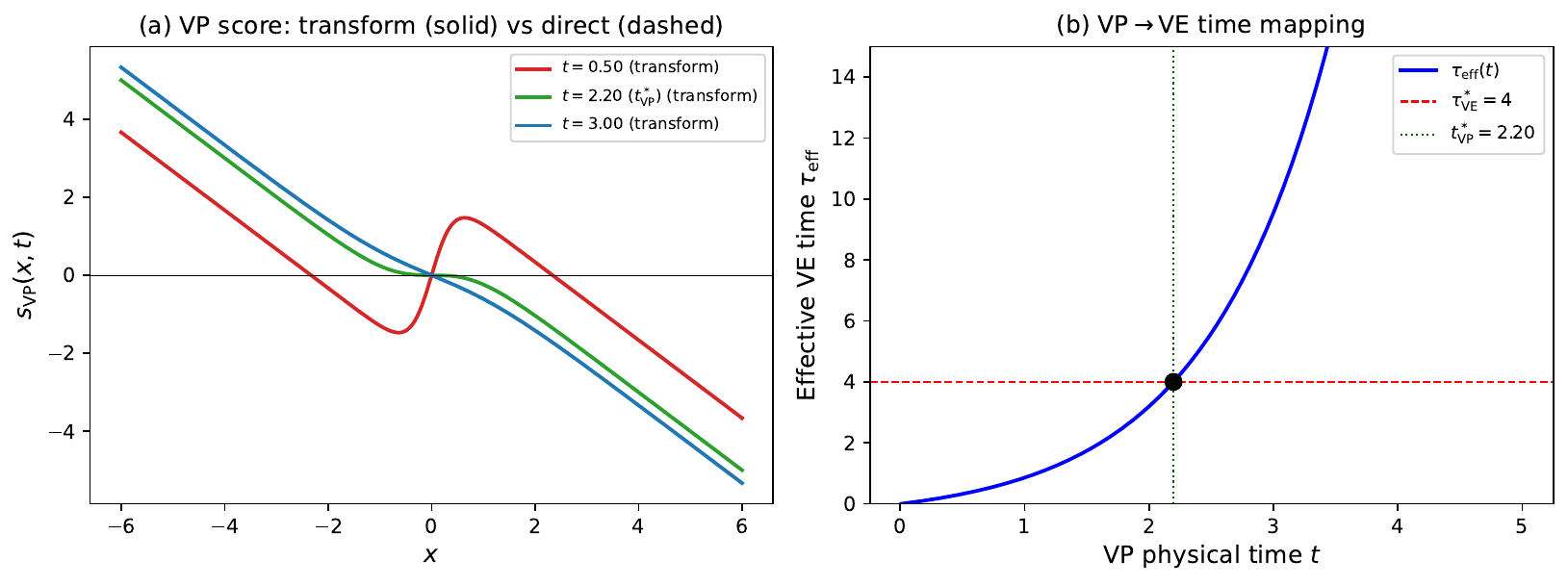}
\caption{VP--VE equivalence under the rescaling transformation. Panel~(a) overlays the VP score from the exact transformation $s_{\mathrm{VP}}(x,t)=\alpha(t)^{-1}s_{\mathrm{VE}}(x/\alpha(t),\tau_{\mathrm{eff}}(t))$ with the score computed directly from the VP marginal; on the plotted times the two are visually indistinguishable. Panel~(b) plots the effective-time map $\tau_{\mathrm{eff}}(t)$ together with the VP speciation time $t_{\mathrm{VP}}^\ast = \ln 9 \approx 2.20$, which lands exactly at the VE critical time $\tau_{\mathrm{VE}}^\ast = 4.0$.}
\label{fig:vp-ve}
\end{figure}

\section{Correction Terms for Asymmetric Mixtures}\label{sec:correction}

The leading-order speciation formula $\tau^\ast_{\mathrm{LO}} = (\lambda_1(\bm{W}) - \sigma_0^2)/2$ of \Cref{thm:d-speciation} becomes exact for symmetric arrangements (equal-weight binary mixtures, regular simplices) but admits corrections for general $K$-component mixtures.
Here we derive these corrections by expanding the posterior responsibilities in powers of $1/\sigma_\tau^2$ and tracing their effect on the score Jacobian.

\subsection{Posterior responsibilities at the weighted mean}\label{sec:corr:responsibilities}

\begin{definition}[Posterior responsibility]\label{def:responsibility}
For the GMM~\eqref{eq:gmm-noised}, the posterior responsibility of component $k$ at point $\bm{x}$ is
\begin{equation}\label{eq:rk-def}
  r_k(\bm{x},\tau) = \frac{w_k\,\mathcal{N}(\bm{x};\,\bm{\mu}_k,\,\sigma_\tau^2\bm{I})}{\sum_{m=1}^K w_m\,\mathcal{N}(\bm{x};\,\bm{\mu}_m,\,\sigma_\tau^2\bm{I})}.
\end{equation}
\end{definition}

At the weighted mean $\bar{\bm{x}} = \sum_k w_k \bm{\mu}_k$, define the squared distances $d_k^2 = |\bar{\bm{x}} - \bm{\mu}_k|^2 = |\bm{\nu}_k|^2$ and the dimensionless parameters $\eta_k = d_k^2/(2\sigma_\tau^2)$.

\begin{proposition}[Responsibility expansion]\label{prop:rk-expansion}
For large $\sigma_\tau^2$ (i.e., $\eta_k \ll 1$), the responsibilities at $\bar{\bm{x}}$ admit the expansion
\begin{equation}\label{eq:rk-expansion}
  r_k(\bar{\bm{x}},\tau)
  = w_k\!\left[1 + (\langle\eta\rangle - \eta_k) + \frac{(\eta_k - \langle\eta\rangle)^2 - \Var_w(\eta)}{2} + O(\eta^3)\right],
\end{equation}
where $\langle\eta\rangle = \sum_m w_m \eta_m$ and $\Var_w(\eta) = \langle\eta^2\rangle - \langle\eta\rangle^2$.
\end{proposition}

\begin{proof}
Write $r_k = w_k e^{-\eta_k}/\sum_m w_m e^{-\eta_m}$.
Expanding $e^{-\eta_k} = 1 - \eta_k + \eta_k^2/2 + O(\eta^3)$:
\[
  \sum_m w_m e^{-\eta_m} = 1 - \langle\eta\rangle + \langle\eta^2\rangle/2 + O(\eta^3).
\]
Dividing and expanding $(1-\epsilon)^{-1} = 1 + \epsilon + \epsilon^2 + \cdots$ with $\epsilon = \langle\eta\rangle - \langle\eta^2\rangle/2 + \cdots$:
\begin{align*}
  r_k &= w_k(1 - \eta_k + \eta_k^2/2)(1 + \langle\eta\rangle - \langle\eta^2\rangle/2 + \langle\eta\rangle^2 + \cdots)\\
      &= w_k\bigl[1 + (\langle\eta\rangle - \eta_k) + \tfrac{1}{2}\bigl((\eta_k - \langle\eta\rangle)^2 - \Var_w(\eta)\bigr) + O(\eta^3)\bigr].
\end{align*}
One verifies $\sum_k r_k = 1$ at each order: order~0 gives $\sum w_k = 1$; order~1 gives $\langle\langle\eta\rangle - \eta\rangle = 0$; order~2 gives $\langle(\eta - \langle\eta\rangle)^2 - \Var(\eta)\rangle/2 = 0$.
\end{proof}

\begin{corollary}[Exactness condition]\label{cor:exactness}
The responsibilities satisfy $r_k = w_k$ exactly if and only if all $\eta_k$ are equal, i.e., all component means are equidistant from $\bar{\bm{x}}$.
This holds for: (a)~$K = 2$ with $w_1 = w_2 = 1/2$; (b)~any~$K$ with equal weights and means forming a regular simplex centered at $\bar{\bm{x}}$.
\end{corollary}

\subsection{The corrected Jacobian}\label{sec:corr:jacobian}

The exact score Jacobian at $\bar{\bm{x}}$ is $J_{ij} = -\delta_{ij}/\sigma_\tau^2 + C_{ij}/\sigma_\tau^4$, where $\bm{C} = \sum_k r_k \bm{\mu}_k\bm{\mu}_k^\top - \tilde{\bm{\mu}}\tilde{\bm{\mu}}^\top$ is the posterior covariance of the means (see the proof of \Cref{thm:d-speciation}).
Substituting the expansion of \Cref{prop:rk-expansion} into $\bm{C}$ requires expanding the posterior mean $\tilde{\bm{\mu}} = \sum_k r_k \bm{\mu}_k$ and second moment $\bm{M} = \sum_k r_k \bm{\mu}_k\bm{\mu}_k^\top$.

The next two results give asymptotic expansions in inverse noise variance. They refine the leading-order higher-dimensional criterion from \Cref{prop:voronoi,thm:d-speciation}; the exact non-perturbative characterization is given later in \Cref{thm:exact-criterion}.

\begin{definition}[Distance-weighted covariance]\label{def:Q}
Define the \emph{distance-weighted covariance}:
\begin{equation}\label{eq:Q-def}
  \bm{Q} = \sum_{k=1}^K w_k\,|\bm{\nu}_k|^2\,\bm{\nu}_k\bm{\nu}_k^\top,
\end{equation}
and the mean squared distance $\langle d^2\rangle = \sum_k w_k |\bm{\nu}_k|^2$.
\end{definition}

\begin{theorem}[Corrected Jacobian]\label{thm:corrected-jacobian}
To second order in $1/\sigma_\tau^2$, the score Jacobian at $\bar{\bm{x}}$ admits the expansion
\begin{equation}\label{eq:corrected-jacobian}
  \bm{J}(\bar{\bm{x}},\tau)
  = -\frac{\bm{I}}{\sigma_\tau^2}
  + \frac{\bm{W}}{\sigma_\tau^4}
  + \frac{\langle d^2\rangle\bm{W} - \bm{Q}}{2\sigma_\tau^6}
  + O(\sigma_\tau^{-8}).
\end{equation}
\end{theorem}

\begin{proof}
We expand the posterior covariance $\bm{C} = \bm{M} - \tilde{\bm{\mu}}\tilde{\bm{\mu}}^\top$ order by order.

\medskip\noindent\textbf{Posterior mean.}
Using \Cref{prop:rk-expansion} at first order:
\[
  \tilde{\mu}_i = \sum_k r_k \mu_{k,i}
  = \bar{x}_i + \sum_k w_k(\langle\eta\rangle - \eta_k)\mu_{k,i} + O(\sigma_\tau^{-4}).
\]
Define $\bm{\xi} = \sum_k w_k|\bm{\nu}_k|^2\bm{\nu}_k$ (the ``third moment'' of the centered means).
A direct computation using $\mu_{k,i} = \bar{x}_i + \nu_{k,i}$ and $\langle\eta\rangle - \eta_k = (\langle d^2\rangle - d_k^2)/(2\sigma_\tau^2)$ gives
\begin{equation}\label{eq:posterior-mean}
  \tilde{\bm{\mu}} = \bar{\bm{x}} - \frac{\bm{\xi}}{2\sigma_\tau^2} + O(\sigma_\tau^{-4}).
\end{equation}

\medskip\noindent\textbf{Posterior second moment.}
At leading order: $M_{ij}^{(0)} = \sum_k w_k\mu_{k,i}\mu_{k,j} = W_{ij} + \bar{x}_i\bar{x}_j$.
The first-order correction, after expanding and using $\mu_{k,i} = \bar{x}_i + \nu_{k,i}$, is:
\begin{equation}\label{eq:M-correction}
  M_{ij}^{(1)} = \frac{1}{2\sigma_\tau^2}\bigl[\langle d^2\rangle W_{ij} - Q_{ij} - \bar{x}_i\xi_j - \bar{x}_j\xi_i\bigr].
\end{equation}

\medskip\noindent\textbf{Product of posterior means.}
From~\eqref{eq:posterior-mean}:
\[
  \tilde{\mu}_i\tilde{\mu}_j = \bar{x}_i\bar{x}_j - \frac{\bar{x}_i\xi_j + \bar{x}_j\xi_i}{2\sigma_\tau^2} + O(\sigma_\tau^{-4}).
\]

\medskip\noindent\textbf{Posterior covariance.}
$C_{ij} = M_{ij}^{(0)} + M_{ij}^{(1)} - \tilde{\mu}_i\tilde{\mu}_j$. The $\bar{x}_i\bar{x}_j$ terms cancel between $M^{(0)}$ and $\tilde{\mu}\tilde{\mu}^\top$; the $\bar{x}\xi$ terms cancel between $M^{(1)}$ and $\tilde{\mu}\tilde{\mu}^\top$:
\begin{equation}\label{eq:C-expansion}
  C_{ij} = W_{ij} + \frac{\langle d^2\rangle W_{ij} - Q_{ij}}{2\sigma_\tau^2} + O(\sigma_\tau^{-4}).
\end{equation}
Substituting into $J_{ij} = -\delta_{ij}/\sigma_\tau^2 + C_{ij}/\sigma_\tau^4$ yields~\eqref{eq:corrected-jacobian}.
\end{proof}

\subsection{The corrected speciation time}\label{sec:corr:speciation}

\begin{definition}[Correction coefficient]\label{def:gamma}
Define $\gamma_1 = \langle d^2\rangle\lambda_1 - \bm{e}_1^\top\bm{Q}\,\bm{e}_1$, where $\bm{e}_1$ is the leading eigenvector of $\bm{W}$.
\end{definition}

\begin{theorem}[Corrected speciation time]\label{thm:corrected-speciation}
Including the first-order correction, the speciation time admits the expansion
\begin{equation}\label{eq:corrected-tau-star}
  \tau^\ast
  = \frac{\lambda_1 - \sigma_0^2}{2}
  + \frac{\gamma_1}{4\lambda_1}
  + O\!\left(\frac{\gamma_1^2}{\lambda_1^3}\right).
\end{equation}
The corresponding quadratic approximation is:
\begin{equation}\label{eq:quadratic-tau-star}
  \sigma_{\tau^\ast}^2
  = \frac{\lambda_1 + \sqrt{\lambda_1^2 + 2\gamma_1}}{2}.
\end{equation}
\end{theorem}

\begin{proof}
The leading Jacobian eigenvalue along $\bm{e}_1$ is
\[
  \lambda_1^{(J)} = -\frac{1}{\sigma_\tau^2} + \frac{\lambda_1}{\sigma_\tau^4} + \frac{\gamma_1}{2\sigma_\tau^6} + O(\sigma_\tau^{-8}).
\]
Setting $\lambda_1^{(J)} = 0$ and multiplying by $\sigma_\tau^6$: $\sigma_\tau^4 - \lambda_1\sigma_\tau^2 - \gamma_1/2 = 0$.
The quadratic formula gives~\eqref{eq:quadratic-tau-star}.
Expanding for small $|\gamma_1|/\lambda_1^2$:
$\sigma_\tau^2 \approx \lambda_1 + \gamma_1/(2\lambda_1)$,
hence $\tau^\ast = (\sigma_{\tau^\ast}^2 - \sigma_0^2)/2 = (\lambda_1 - \sigma_0^2)/2 + \gamma_1/(4\lambda_1)$.
\end{proof}

\begin{proposition}[When the correction is negative]\label{prop:gamma-sign}
Let $a_k = \bm{\nu}_k \cdot \bm{e}_1$ and $b_k^2 = |\bm{\nu}_k - a_k\bm{e}_1|^2$, so that $d_k^2 = a_k^2 + b_k^2$.
If $\Cov_w(b^2,a^2) \ge 0$, then $\gamma_1 \le 0$.
In particular, $\gamma_1 \le 0$ for collinear configurations ($b_k \equiv 0$) and whenever all $b_k$ are equal.
In general, however, $\gamma_1$ can have either sign.
\end{proposition}

\begin{proof}
Using $d_k^2 = a_k^2 + b_k^2$ and the covariance identity,
\[
  \gamma_1 = \langle d^2\rangle\langle a^2\rangle - \langle d^2 a^2\rangle
  = -\Cov_w(d^2,a^2)
  = -\Var_w(a^2) - \Cov_w(b^2,a^2).
\]
If $\Cov_w(b^2,a^2) \ge 0$, then the right-hand side is non-positive, proving $\gamma_1 \le 0$.
If the configuration is collinear, then $b_k \equiv 0$, so $\Cov_w(b^2,a^2)=0$.
If all $b_k$ are equal, then $b^2$ is constant and again $\Cov_w(b^2,a^2)=0$.
No sign conclusion is possible without an additional geometric assumption on $\Cov_w(b^2,a^2)$.
\end{proof}

\begin{remark}[Physical interpretation]\label{rem:correction-phys}
When $\gamma_1 < 0$, components closer to $\bar{\bm{x}}$ receive higher posterior responsibility than their prior weight $w_k$.
This biases the posterior covariance toward the closer components, reducing the effective between-class variance and causing speciation to occur at a lower noise level (earlier in the reverse process).
If modes with smaller projection onto $\bm{e}_1$ have sufficiently large perpendicular spread, then $\Cov_w(b^2,a^2)$ can be negative and $\gamma_1$ can instead be positive, delaying the transition.
The correction vanishes for symmetric arrangements where all $d_k$ are equal.
\end{remark}

The numerical effect of this correction is summarized in \Cref{fig:corrections}: the first-order term dramatically improves the speciation-time estimate, and for the asymmetric family plotted there the coefficient $\gamma_1$ remains negative across the tested separations.

\begin{figure}[t]
\centering
\includegraphics[width=\linewidth]{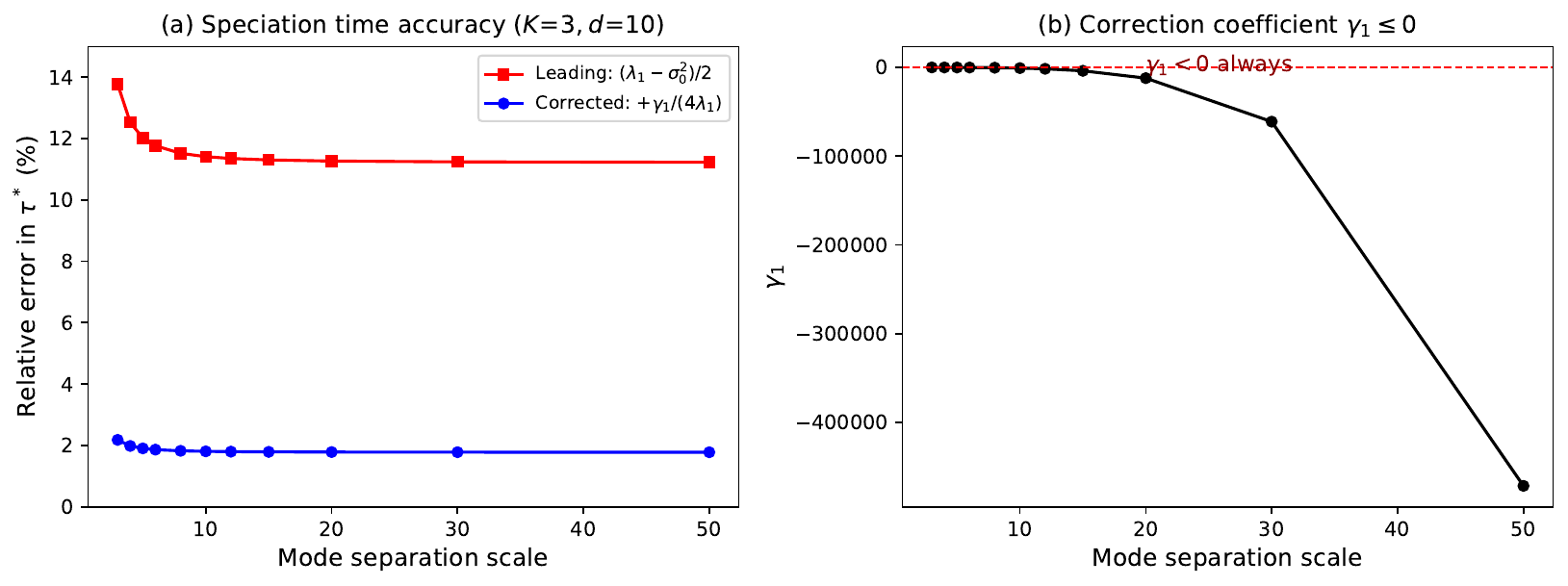}
\caption{Correction terms for an asymmetric mixture. Panel~(a) compares the relative error of the leading-order and corrected speciation formulas as the separation scale increases for a $K=3$, $d=10$ example; the correction cuts the error from about $11\%$ to about $2\%$. Panel~(b) traces the coefficient $\gamma_1$ across the same family. In these examples it stays negative, so the first correction moves the threshold earlier.}
\label{fig:corrections}
\end{figure}

\subsection{The exact non-perturbative criterion}\label{sec:corr:exact}

The exact criterion behind these asymptotic formulas is the following.

\begin{theorem}[Exact speciation criterion]\label{thm:exact-criterion}
The speciation time $\tau^\ast$ is characterized as the unique solution of
\begin{equation}\label{eq:exact-criterion}
  \lambda_{\max}\!\bigl(\bm{C}(\tau)/\sigma_\tau^2\bigr) = 1,
\end{equation}
where $\bm{C}(\tau) = \sum_k r_k(\bar{\bm{x}},\tau)\,\tilde{\bm{\nu}}_k\tilde{\bm{\nu}}_k^\top$ is the exact posterior covariance of the means and $\tilde{\bm{\nu}}_k = \bm{\mu}_k - \tilde{\bm{\mu}}(\tau)$.
This equation can be solved by bisection at cost $O(K^2 d)$ per step.
\end{theorem}

\begin{proof}
The Jacobian eigenvalue is zero iff $\lambda_{\max}(\bm{C})/\sigma_\tau^4 = 1/\sigma_\tau^2$, i.e., $\lambda_{\max}(\bm{C}/\sigma_\tau^2) = 1$.
Uniqueness follows from the monotonicity of $\lambda_{\max}(\bm{C}(\tau)/\sigma_\tau^2)$ in $\tau$: as $\tau$ increases, $\sigma_\tau^2$ grows, the responsibilities become more uniform, and $\bm{C} \to \bm{W}$, while the division by $\sigma_\tau^2$ shrinks the eigenvalue, ensuring a unique crossing.
\end{proof}

\begin{remark}[Hierarchy of formulas]\label{rem:hierarchy}
There are really three levels here, depending on how much simplicity one wants and how much accuracy one needs:
\begin{enumerate}[nosep]
\item \emph{Exact} (\Cref{thm:exact-criterion}): always valid, cost $O(K^2 d)$ per bisection step.
\item \emph{Closed-form, exact for symmetric} (\Cref{thm:shock-speciation}): $\tau^\ast = (\lambda_1(\bm{W}) - \sigma_0^2)/2$.
\item \emph{Closed-form with correction} (\Cref{thm:corrected-speciation}): includes $\gamma_1/(4\lambda_1)$; error ${\sim}2\%$ for equal-weight asymmetric mixtures with moderate separation.
\end{enumerate}
\end{remark}

\section{Numerical Verification}\label{sec:experiments}

The experiments are grouped in a fairly plain way: first the Burgers PDE and the closed-form Gaussian calculations, then the higher-dimensional, VP, and correction results, and finally the quartic-well test of the local theorem.

\subsection{Verification of the score PDE (Theorems~\ref{thm:score-pde} and~\ref{thm:score-burgers})}\label{sec:exp:pde}

For the symmetric binary GMM~\eqref{eq:binary-gmm} with $a = 3$, $\sigma_0 = 1$, we evaluate the exact score~\eqref{eq:exact-score} and its temporal and spatial derivatives on a grid of $2000$ points $x \in [-8,8]$, using finite differences ($\Delta\tau = 10^{-7}$) for the time derivative.
The residual curves are shown in \Cref{fig:pde-verification}.

\paragraph{Score PDE.}
We compute the residual $|s_\tau - (s_{xx} + 2\,s\,s_x)|$ at five values of $\tau$.
The maximum pointwise error is below $5 \times 10^{-9}$ at all times tested, confirming \Cref{thm:score-pde} to machine precision.

\paragraph{Burgers equation.}
Transforming to $u = -2s$, the residual $|u_\tau + u\,u_x - u_{xx}|$ is below $9 \times 10^{-9}$ at all times, confirming \Cref{thm:score-burgers}.

\subsection{Verification of the speciation time (Theorem~\ref{thm:shock-speciation})}\label{sec:exp:speciation}

We compute $s_x(0,\tau)$ from~\eqref{eq:sx-origin} for $\tau \in [0.1,\, 10]$ and locate its zero crossing numerically.
The zero crossing and the width trend are summarized in \Cref{fig:speciation-width}.

\paragraph{Result.}
The predicted speciation time is $\tau^\ast = (9 - 1)/2 = 4.0$.
The numerical zero crossing of $s_x(0,\tau)$ occurs at $\tau = 4.0000$ (to four decimal places), with error $< 5 \times 10^{-5}$.
The analytical formula~\eqref{eq:sx-origin} gives $s_x(0,\tau^\ast) = 0$.

\subsection{Verification of the interfacial profile (Proposition~\ref{prop:shock-profile})}\label{sec:exp:profile}

The predicted interfacial width is $\delta(\tau) = \sigma_\tau^2/a$.
Representative profiles are displayed in \Cref{fig:score-burgers}, and the width law is plotted in \Cref{fig:speciation-width}.
At $\tau = 0.1$: $\delta = 1.2/3 = 0.4$.
At $\tau = 0.5$: $\delta = 2/3 \approx 0.667$.
At $\tau = 1.0$: $\delta = 3/3 = 1.0$.
In each case, the width of the background-subtracted $\tanh$ profile in~\eqref{eq:background-subtracted-score} matches the prediction.

\subsection{Verification of the beyond-Gaussian local theorem on a quartic well}\label{sec:exp:quartic}

To test \Cref{thm:local-tanh,thm:local-speciation} beyond Gaussian mixtures, we consider the quartic double-well density
\begin{equation}\label{eq:quartic-well}
  p_0(x) \propto \exp\!\left(-\frac{(x^2-a^2)^2}{4a^2}\right),
\end{equation}
whose tails are quartic rather than Gaussian.
We split the initial density into the left and right attraction basins, convolve each part with the heat kernel by numerical quadrature, and evaluate the exact decomposition \eqref{eq:local-tanh} and boundary criterion \eqref{eq:local-normal-criterion} directly.

\paragraph{Exact local decomposition.}
The identity \eqref{eq:local-tanh} is satisfied to within ${\sim}10^{-4}$ uniformly on the tested grid; the residual is limited by the quadrature used to evaluate the heat-kernel integrals rather than by the theorem itself.

\paragraph{Exact speciation criterion.}
Solving the boundary equation numerically gives a speciation time $\tau^\ast \approx 1.948$ for this quartic well.
At that time, the residual in the exact normal criterion
\[
  \partial_n s_n = \partial_n \bar s_n + \kappa^2/4
\]
is $5.6\times 10^{-5}$.
A naive matched-Gaussian estimate would instead predict $\tau^\ast = 3.0$, so the local theorem captures non-Gaussian boundary geometry that is missed by a Gaussian-mixture proxy.

\subsection{Verification of the amplification exponent (Theorem~\ref{thm:amplification})}\label{sec:exp:amplification}

We compare the closed-form exponent $\Lambda(\tau) = \tfrac{1}{2}[a^2/\sigma_\tau^2 - 1 - \ln(a^2/\sigma_\tau^2)]$ against numerical integration of $\int_\tau^{\tau^\ast} s_x(0,\tau')\,d\tau'$ using the trapezoidal rule with $N = 10{,}000$ points.
The amplification curve and representative reverse-time trajectories are shown in \Cref{fig:amplification-trajectories}.

\begin{center}
\begin{tabular}{ccccl}
\toprule
$\tau$ & $\Lambda_{\text{exact}}$ & $\Lambda_{\text{numerical}}$ & Error & Amplification $e^\Lambda$ \\
\midrule
$0.0$ & $2.9014$ & $2.9014$ & $4.5 \times 10^{-7}$ & $18.2\times$ \\
$0.5$ & $0.9980$ & $0.9980$ & $4.1 \times 10^{-8}$ & $2.7\times$ \\
$1.0$ & $0.4507$ & $0.4507$ & $8.2 \times 10^{-9}$ & $1.6\times$ \\
$2.0$ & $0.1061$ & $0.1061$ & $6.1 \times 10^{-10}$ & $1.1\times$ \\
\bottomrule
\end{tabular}
\end{center}

The closed form matches the numerical integral to at least seven significant figures.

\subsection{Verification of curl preservation (Theorem~\ref{thm:curl-preservation})}\label{sec:exp:curl}

For a two-component GMM in $d = 2$ with asymmetric means $\bm{\mu}_1 = (2,1)$, $\bm{\mu}_2 = (-1, 1.5)$ and weights $w_1 = 0.4$, $w_2 = 0.6$, we compute the curl $\omega = \partial_1 s_2 - \partial_2 s_1$ at $200$ random points using centered finite differences ($\Delta x = 10^{-6}$).
The associated quiver plot and curl summary appear in \Cref{fig:2d-score-curl}.

\paragraph{Result.}
For every tested noise level $\tau \in \{0.1, 0.5, 1.0, 3.0, 10.0\}$, the maximum curl magnitude is below $1.3 \times 10^{-9}$.
This confirms \Cref{thm:curl-preservation} to machine precision.

\subsection{Verification of the VP--VE equivalence (Theorem~\ref{thm:vp-ve})}\label{sec:exp:vp}

For the symmetric binary GMM with $a = 3$, $\sigma_0 = 1$, $\beta = 1$:
The numerical comparison is displayed in \Cref{fig:vp-ve}.

\paragraph{Speciation time.}
The predicted VP speciation time is $t^\ast = \ln(a^2 - \sigma_0^2 + 1) = \ln 9 \approx 2.197$.
The effective VE time at this point is $\tau_{\mathrm{eff}}(t^\ast) = (1 - e^{-\ln 9})/(2e^{-\ln 9}) = (1 - 1/9)/(2/9) = (8/9)/(2/9) = 4.000$, matching $\tau^\ast_{\mathrm{VE}} = 4.0$.

\paragraph{Score transformation.}
We compare $s_{\mathrm{VP}}(x,t) = \alpha^{-1}\,s_{\mathrm{VE}}(x/\alpha, \tau_{\mathrm{eff}})$ against the direct VP score (computed from the VP marginal $\tfrac{1}{2}\mathcal{N}(x;\, -\alpha a,\, \sigma_{\mathrm{VP}}^2) + \tfrac{1}{2}\mathcal{N}(x;\, \alpha a,\, \sigma_{\mathrm{VP}}^2)$).
At five values of $t$, the maximum pointwise discrepancy over $x \in [-5,5]$ is below $2 \times 10^{-15}$---consistent with double-precision arithmetic.

\subsection{Verification of correction terms (Theorems~\ref{thm:corrected-jacobian} and~\ref{thm:corrected-speciation})}\label{sec:exp:correction}

We test the correction on three configurations.
The aggregate error reduction and the sign of $\gamma_1$ are shown in \Cref{fig:corrections}.

\paragraph{Symmetric case: $K = 2$, $d = 5$, equal weights.}
With $\bm{\mu}_1 = (3,1,0.5,0,0)$ and $\bm{\mu}_2 = -\bm{\mu}_1$:
the predicted $\gamma_1 = 0$ (by \Cref{cor:exactness}), and the leading-order speciation matches the exact value to $10^{-6}$ relative error.

\paragraph{Asymmetric case: $K = 3$, $d = 10$, equal weights.}
We compare leading-order, first-correction, and quadratic~\eqref{eq:quadratic-tau-star} predictions against the exact criterion~\eqref{eq:exact-criterion} solved by bisection.

\begin{center}
\begin{tabular}{ccccc}
\toprule
Separation scale & $\tau^\ast_{\text{leading}}$ & $\tau^\ast_{\text{corrected}}$ & $\tau^\ast_{\text{exact}}$ & Leading / Corrected error \\
\midrule
$5$  & $7.83$ & $7.12$ & $6.99$ & $12.0\%$ / $1.9\%$ \\
$12$ & $47.5$ & $43.4$ & $42.6$ & $11.4\%$ / $1.8\%$ \\
$50$ & $833$  & $762$  & $749$  & $11.2\%$ / $1.8\%$ \\
\bottomrule
\end{tabular}
\end{center}

The first-order correction reduces the error from ${\sim}11\%$ to ${\sim}2\%$; the quadratic formula further reduces it to ${\sim}0.8\%$.
For the asymmetric families tested here, the sign of $\gamma_1$ is negative in all cases.

\paragraph{Equilateral case: $K = 3$, $d = 2$.}
With means at the vertices of an equilateral triangle of circumradius $R = 4$:
$\gamma_1 = 0$ to machine precision, and the leading-order formula is exact.

\section{Conclusion}\label{sec:conclusion}

\subsection{Summary of contributions}

The main point of the paper is simple: the score function of a diffusion generative model satisfies a viscous Burgers equation, with cumulative noise variance playing the role of viscosity. From there the paper moves through the PDE correspondence, the local boundary theorem, and the Gaussian formulas built on top of it. The identification itself is a direct consequence of the classical Cole--Hopf transform~\citep{hopf1950partial, cole1951quasi} applied to the heat equation governing the forward diffusion~\citep{sohl2015deep, song2021score}. The main consequences are as follows:

\begin{enumerate}[label=(\roman*),leftmargin=2em,nosep]
\item \textbf{Local binary-boundary theorem.}
For any decomposition of the noised density into two positive heat solutions, the score splits exactly as $\bm{s} = \bar{\bm{s}} + \tfrac{1}{2}\tanh(\phi/2)\nabla\phi$ (\Cref{thm:local-tanh}), and at any regular binary boundary the normal Hessian obeys the exact criterion $\partial_n s_n = \partial_n \bar s_n + \kappa^2/4$ (\Cref{thm:local-speciation}).

\item \textbf{Gaussian specialization and spectral match.}
For symmetric binary mixtures, the local criterion reduces to the midpoint-derivative condition and coincides exactly with the spectral criterion of \citet{raya2023spontaneous} and \citet{biroli2024dynamical} (\Cref{thm:shock-speciation}). In higher-dimensional asymmetric Gaussian settings, the leading-order threshold is $\tau^\ast_{\mathrm{LO}} = (\lambda_1(\bm{W}) - \sigma_0^2)/2$, with explicit correction terms and an exact non-perturbative refinement given in \Cref{thm:d-speciation,thm:corrected-speciation,thm:exact-criterion}.

\item \textbf{Interfacial profile.}
After subtracting the smooth background drift, the inter-mode layer is locally a Burgers $\tanh$ profile (\Cref{thm:local-speciation}); in the symmetric Gaussian case this profile is globally exact (\Cref{prop:shock-profile}) with width $\delta = \sigma_\tau^2/a$.

\item \textbf{Error amplification.}
Score estimation errors are amplified near mode boundaries by $\exp(\Lambda)$ where $\Lambda = \tfrac{1}{2}[\mathrm{SNR} - 1 - \ln\mathrm{SNR}] \approx \mathrm{SNR}/2$ (\Cref{thm:amplification}), providing a PDE-theoretic explanation for the empirical sensitivity of diffusion models to low-noise score accuracy~\citep{song2020improved, karras2022elucidating}.

\item \textbf{Curl preservation.}
The vector Burgers dynamics preserves irrotationality (\Cref{thm:curl-preservation}), establishing that the non-conservative scores documented by \citet{vuong2025score} cannot arise from the exact score dynamics alone.

\item \textbf{VP--VE unification.}
The coordinate transformation $Z = X/\alpha(t)$ reduces the VP-SDE to the VE case (\Cref{thm:vp-ve}), yielding closed-form VP speciation times and interfacial widths (\Cref{cor:vp-speciation,cor:vp-shock}).

\item \textbf{Decision boundary dynamics.}
The Rankine--Hugoniot condition governs the motion of mode boundaries for asymmetric mixtures (\Cref{prop:rh}), and the scalar Lax entropy condition provides a diagnostic on one-dimensional boundary slices (\Cref{prop:entropy}).
\end{enumerate}

All results are proved in the text. The Gaussian-mixture formulas are verified to machine precision (${\sim}10^{-9}$), and the local beyond-Gaussian theorem is also checked on a quartic double-well.

\subsection{Implications for practice}

A few remarks about diffusion model design are worth keeping in view:

\paragraph{Adaptive step-size schedules.}
The error amplification exponent (\Cref{thm:amplification}) provides a principled signal for allocating ODE solver steps: the step size should scale inversely with $|s_x|$, concentrating discretization effort near the interfacial layer (mode boundary) and below the speciation time $\tau^\ast$.
This recovers---and gives a theoretical justification for---the empirical observation that low-noise regions require finer discretization~\citep{karras2022elucidating, song2021ddim}.

\paragraph{Score network diagnostics.}
The one-dimensional Lax entropy condition on normal slices (\Cref{prop:entropy}) and curl-freeness (\Cref{thm:curl-preservation}) provide checkable constraints on learned scores.
A score network that violates the scalar entropy condition on a boundary-normal slice, or that has large curl there, is likely to produce poor samples near that boundary.
The ``score Fokker--Planck'' regularizer of \citet{lai2023fp}---which, as we have shown, enforces the Burgers equation---can be understood as implicitly penalizing entropy-violating scalar slices.

\paragraph{Noise schedule design.}
The VP--VE reduction (\Cref{thm:vp-ve}) shows that noise schedule optimization for VP models can be conducted entirely in the effective VE time $\tau_{\mathrm{eff}}(t)$, reducing the design problem to choosing how the schedule traverses the interfacial layer.

\subsection{Limitations and open problems}

\paragraph{Beyond Gaussian mixtures.}
The local binary-boundary theorem of \Cref{thm:local-tanh,thm:local-speciation} is already exact for arbitrary smooth densities once a two-component heat decomposition is specified, and \Cref{sec:exp:quartic} confirms this on a non-Gaussian quartic well.
One open problem is to obtain comparably explicit formulas for the background field $\bar{\bm{s}}$, the log-ratio gradient $\kappa$, and the resulting speciation time in non-Gaussian settings; outside the Gaussian case these quantities typically have to be computed numerically. A separate issue is the binary-reduction assumption itself when more than two modes compete. \Cref{prop:binary-remainder} shows that the error is exponentially small for well-separated binary boundaries, but triple junctions and strongly non-local mode interactions are still missing from the present analysis.

\paragraph{The role of architecture.}
Our analysis assumes access to the true score or a pointwise approximation thereof.
Understanding how specific neural network architectures (e.g., U-Nets) interact with the Burgers interfacial structure---whether they introduce systematic biases toward or away from entropy-satisfying solutions---remains open.
The observation by \citet{vuong2025score} that trained networks produce non-conservative fields suggests that architecture imposes an implicit Helmholtz decomposition~\citep{bhatia2013helmholtz} on the learned score. That curl component deserves a more systematic treatment through the Burgers framework than we have given here.

\paragraph{Higher-order corrections.}
The correction series of \Cref{sec:correction} was carried to first order ($\gamma_1$).
Computing the $O(\sigma_\tau^{-8})$ term would tighten the approximation for strongly asymmetric mixtures; the algebraic structure of the expansion (powers of the responsibility deviation $\eta_k - \langle\eta\rangle$) is systematic and could in principle be automated.

\paragraph{Multi-dimensional shocks.}
In $d > 1$, the formal inviscid vector Burgers description develops shock \emph{surfaces} (\Cref{prop:voronoi}).
The detailed structure of these surfaces---their curvature, their interaction at triple junctions where three Voronoi cells meet, and the associated Rankine--Hugoniot dynamics in $\R^d$---is largely unexplored in the generative modeling context and connects to the rich mathematical theory of multi-dimensional conservation laws~\citep{evans2010partial}.

\paragraph{Stochastic corrections.}
The probability flow ODE is the deterministic counterpart of the reverse SDE~\eqref{eq:reverse-sde}.
The stochastic term in the reverse SDE introduces a viscous regularization that smooths the interfacial layers, analogous to adding viscosity.
It would also be worthwhile to quantify the interplay between stochasticity, score error, and interfacial structure, perhaps through the stochastic localization framework~\citep{montanari2023sampling, benton2024nearly}.

\subsection{Closing remark}

The Burgers equation was introduced by \citet{burgers1948mathematical} in 1948 as a toy model for turbulence.
Diffusion generative models were introduced by \citet{sohl2015deep} in 2015 as a new approach to density estimation.
The connection between the two is a direct consequence of the Cole--Hopf transform~\citep{hopf1950partial, cole1951quasi} applied to the heat equation.
Making this link explicit clarifies the role of interfacial structure, error amplification, and boundary dynamics in diffusion models.

\bibliographystyle{plainnat}
\bibliography{References3}

@article{achilli2025memorization,
 archivePrefix = {arXiv},
 author = {Achilli, Beatrice and Ambrogioni, Luca and Lucibello, Carlo and M{\'e}zard, Marc and Ventura, Emanuela},
 doi = {10.48550/arXiv.2502.09578},
 eprint = {2502.09578},
 journal = {{Journal of Statistical Mechanics: Theory and Experiment}},
 pages = {073401},
 title = {Memorization and Generalization in Generative Diffusion under the Manifold Hypothesis},
 url = {https://arxiv.org/abs/2502.09578},
 volume = {2025},
 year = {2025}
}

@article{albergo2023stochastic,
 archivePrefix = {arXiv},
 author = {Albergo, Michael S and Boffi, Nicholas M and Vanden-Eijnden, Eric},
 doi = {10.48550/arXiv.2303.08797},
 eprint = {2303.08797},
 journal = {{arXiv preprint}},
 title = {Stochastic interpolants: A unifying framework for flows and diffusions},
 url = {https://arxiv.org/abs/2303.08797},
 year = {2023}
}

@article{ambrogioni2025information,
 author = {Ambrogioni, Luca},
 doi = {10.3390/e28020195},
 journal = {{Entropy}},
 title = {The Information Dynamics of Generative Diffusion},
 url = {https://doi.org/10.3390/e28020195},
 volume = {27},
 year = {2025}
}

@article{ambrogioni2025thermodynamics,
 author = {Ambrogioni, Luca},
 doi = {10.3390/e27030291},
 journal = {{Entropy}},
 number = {3},
 pages = {291},
 title = {The Statistical Thermodynamics of Generative Diffusion Models: Phase Transitions, Symmetry Breaking and Critical Instability},
 url = {https://doi.org/10.3390/e27030291},
 volume = {27},
 year = {2025}
}

@book{ambrosio2005gradient,
 author = {Ambrosio, Luigi and Gigli, Nicola and Savar{\'e}, Giuseppe},
 publisher = {Birkh{\"a}user},
 title = {Gradient Flows in Metric Spaces and in the Space of Probability Measures},
 year = {2005}
}

@article{anderson1982reverse,
 author = {Anderson, Brian DO},
 journal = {Stochastic Processes and their Applications},
 number = {3},
 pages = {313--326},
 publisher = {Elsevier},
 title = {Reverse-time diffusion equation models},
 volume = {12},
 year = {1982}
}

@article{benton2024nearly,
 archivePrefix = {arXiv},
 author = {Benton, Joe and De Bortoli, Valentin and Doucet, Arnaud and Durmus, Alain},
 doi = {10.48550/arXiv.2308.03686},
 eprint = {2308.03686},
 journal = {{International Conference on Learning Representations (ICLR)}},
 title = {Nearly $d$-Linear Convergence Bounds for Diffusion Models via Stochastic Localization},
 url = {https://arxiv.org/abs/2308.03686},
 year = {2024}
}

@article{bhatia2013helmholtz,
 author = {Bhatia, Harsh and Norgard, Gregory and Pascucci, Valerio and Bremer, Peer-Timo},
 journal = {IEEE Transactions on Visualization and Computer Graphics},
 number = {8},
 pages = {1386--1404},
 title = {The {{H}}elmholtz-{{H}}odge Decomposition -- A Survey},
 volume = {19},
 year = {2013}
}

@article{biroli2023generative,
 author = {Biroli, Giulio and M{\'e}zard, Marc},
 doi = {10.1088/1742-5468/acf8ba},
 journal = {{Journal of Statistical Mechanics: Theory and Experiment}},
 pages = {093402},
 title = {Generative diffusion in very large dimensions},
 url = {https://doi.org/10.1088/1742-5468/acf8ba},
 volume = {2023},
 year = {2023}
}

@article{biroli2024dynamical,
 author = {Biroli, Giulio and Bonnaire, Tony and de Bortoli, Valentin and M{\'e}zard, Marc},
 doi = {10.1038/s41467-024-54281-3},
 journal = {{Nature Communications}},
 pages = {9957},
 title = {Dynamical Regimes of Diffusion Models},
 url = {https://doi.org/10.1038/s41467-024-54281-3},
 volume = {15},
 year = {2024}
}

@article{bonnaire2025memorization,
 archivePrefix = {arXiv},
 author = {Bonnaire, Tony and Urfin, Rapha{\"e}l and Biroli, Giulio and M{\'e}zard, Marc},
 doi = {10.48550/arXiv.2505.17638},
 eprint = {2505.17638},
 journal = {{arXiv preprint}},
 title = {Why Diffusion Models Don't Memorize: The Role of Implicit Dynamical Regularization in Training},
 url = {https://arxiv.org/abs/2505.17638},
 year = {2025}
}

@article{burgers1948mathematical,
 author = {Burgers, Johannes Martinus},
 doi = {10.1016/S0065-2156(08)70100-8},
 journal = {{Advances in Applied Mechanics}},
 pages = {171--199},
 publisher = {Elsevier},
 title = {A mathematical model illustrating the theory of turbulence},
 url = {https://doi.org/10.1016/S0065-2156(08)70100-8},
 volume = {1},
 year = {1948}
}

@article{chen2023sampling,
 archivePrefix = {arXiv},
 author = {Chen, Sitan and Chewi, Sinho and Li, Jerry and Li, Yuanzhi and Salim, Adil and Zhang, Anru R},
 doi = {10.48550/arXiv.2209.11215},
 eprint = {2209.11215},
 journal = {{International Conference on Learning Representations (ICLR)}},
 title = {Sampling is as easy as learning the score: theory for diffusion models with minimal data assumptions},
 url = {https://arxiv.org/abs/2209.11215},
 year = {2023}
}

@article{cole1951quasi,
 author = {Cole, Julian D},
 doi = {10.1090/qam/42889},
 journal = {{Quarterly of Applied Mathematics}},
 number = {3},
 pages = {225--236},
 title = {On a quasi-linear parabolic equation occurring in aerodynamics},
 url = {https://doi.org/10.1090/qam/42889},
 volume = {9},
 year = {1951}
}

@article{debortoli2022convergence,
 archivePrefix = {arXiv},
 author = {De Bortoli, Valentin},
 doi = {10.48550/arXiv.2208.05314},
 eprint = {2208.05314},
 journal = {{Transactions on Machine Learning Research}},
 title = {Convergence of denoising diffusion models under the manifold hypothesis},
 url = {https://arxiv.org/abs/2208.05314},
 year = {2022}
}

@article{dhariwal2021diffusion,
 archivePrefix = {arXiv},
 author = {Dhariwal, Prafulla and Nichol, Alexander},
 doi = {10.48550/arXiv.2105.05233},
 eprint = {2105.05233},
 journal = {{Advances in Neural Information Processing Systems (NeurIPS)}},
 pages = {8780--8794},
 title = {Diffusion models beat {{GAN}}s on image synthesis},
 url = {https://arxiv.org/abs/2105.05233},
 volume = {34},
 year = {2021}
}

@inproceedings{elalaoui2022sampling,
 author = {El Alaoui, Ahmed and Montanari, Andrea and Sellke, Mark},
 booktitle = {{Proceedings of IEEE FOCS}},
 doi = {10.1109/FOCS54488.2022.00039},
 pages = {323--334},
 title = {Sampling from the {{S}}herrington-{{K}}irkpatrick {{G}}ibbs measure via algorithmic stochastic localization},
 url = {https://doi.org/10.1109/FOCS54488.2022.00039},
 year = {2022}
}

@book{evans2010partial,
 author = {Evans, Lawrence C},
 doi = {10.1090/gsm/191},
 edition = {2nd},
 publisher = {American Mathematical Society},
 title = {Partial Differential Equations},
 url = {https://doi.org/10.1090/gsm/191},
 year = {2010}
}

@article{girsanov1960transforming,
 author = {Girsanov, Igor V},
 journal = {Theory of Probability and its Applications},
 number = {3},
 pages = {285--301},
 title = {On transforming a certain class of stochastic processes by absolutely continuous substitution of measures},
 volume = {5},
 year = {1960}
}

@article{gronwall1919note,
 author = {Gr{\"o}nwall, Thomas Hakon},
 journal = {Annals of Mathematics},
 number = {4},
 pages = {292--296},
 title = {Note on the derivatives with respect to a parameter of the solutions of a system of differential equations},
 volume = {20},
 year = {1919}
}

@article{ho2020denoising,
 archivePrefix = {arXiv},
 author = {Ho, Jonathan and Jain, Ajay and Abbeel, Pieter},
 doi = {10.48550/arXiv.2006.11239},
 eprint = {2006.11239},
 journal = {{Advances in Neural Information Processing Systems (NeurIPS)}},
 pages = {6840--6851},
 title = {Denoising Diffusion Probabilistic Models},
 url = {https://arxiv.org/abs/2006.11239},
 volume = {33},
 year = {2020}
}

@article{ho2022video,
 archivePrefix = {arXiv},
 author = {Ho, Jonathan and Salimans, Tim and Gritsenko, Alexey and Chan, William and Norouzi, Mohammad and Fleet, David J},
 doi = {10.48550/arXiv.2204.03458},
 eprint = {2204.03458},
 journal = {{Advances in Neural Information Processing Systems (NeurIPS)}},
 title = {Video Diffusion Models},
 url = {https://arxiv.org/abs/2204.03458},
 volume = {35},
 year = {2022}
}

@article{hopf1950partial,
 author = {Hopf, Eberhard},
 doi = {10.1002/cpa.3160030302},
 journal = {{Communications on Pure and Applied Mathematics}},
 number = {3},
 pages = {201--230},
 title = {The partial differential equation $u_t + u u_x = {\mu} u_{{xx}}$},
 url = {https://doi.org/10.1002/cpa.3160030302},
 volume = {3},
 year = {1950}
}

@article{hugoniot1889propagation,
 author = {Hugoniot, Pierre-Henri},
 doi = {10.24033/bsmf.350},
 journal = {{Journal de l''Ecole Polytechnique}},
 pages = {1--125},
 title = {Sur la propagation du mouvement dans les corps et sp{{\'e}}cialement dans les gaz parfaits},
 url = {https://doi.org/10.24033/bsmf.350},
 volume = {58},
 year = {1889}
}

@article{hyvarinen2005estimation,
 author = {Hyv{\"a}rinen, Aapo},
 journal = {Journal of Machine Learning Research},
 pages = {695--709},
 title = {Estimation of non-normalized statistical models by score matching},
 volume = {6},
 year = {2005}
}

@article{jordan1998variational,
 author = {Jordan, Richard and Kinderlehrer, David and Otto, Felix},
 journal = {SIAM Journal on Mathematical Analysis},
 number = {1},
 pages = {1--17},
 title = {The variational formulation of the {{F}}okker-{{P}}lanck equation},
 volume = {29},
 year = {1998}
}

@article{kardar1986dynamic,
 author = {Kardar, Mehran and Parisi, Giorgio and Zhang, Yi-Cheng},
 journal = {Physical Review Letters},
 number = {9},
 pages = {889},
 title = {Dynamic scaling of growing interfaces},
 volume = {56},
 year = {1986}
}

@article{karras2022elucidating,
 archivePrefix = {arXiv},
 author = {Karras, Tero and Aittala, Miika and Aila, Timo and Laine, Samuli},
 doi = {10.48550/arXiv.2206.00364},
 eprint = {2206.00364},
 journal = {{Advances in Neural Information Processing Systems (NeurIPS)}},
 title = {Elucidating the Design Space of Diffusion-Based Generative Models},
 url = {https://arxiv.org/abs/2206.00364},
 volume = {35},
 year = {2022}
}

@article{kingma2021variational,
 archivePrefix = {arXiv},
 author = {Kingma, Diederik and Salimans, Tim and Poole, Ben and Ho, Jonathan},
 doi = {10.48550/arXiv.2107.00630},
 eprint = {2107.00630},
 journal = {{Advances in Neural Information Processing Systems (NeurIPS)}},
 title = {Variational Diffusion Models},
 url = {https://arxiv.org/abs/2107.00630},
 volume = {34},
 year = {2021}
}

@inproceedings{lai2023fp,
 archivePrefix = {arXiv},
 author = {Lai, Chieh-Hsin and Takida, Yuhta and Murata, Naoki and Uesaka, Toshimitsu and Mitsufuji, Yuki and Ermon, Stefano},
 booktitle = {{International Conference on Machine Learning (ICML)}},
 doi = {10.48550/arXiv.2210.04296},
 eprint = {2210.04296},
 title = {{{fp}}-diffusion: Improving Score-based Diffusion Models by Enforcing the Underlying Score {{F}}okker-{{P}}lanck Equation},
 url = {https://arxiv.org/abs/2210.04296},
 year = {2023}
}

@article{lax1957hyperbolic,
 author = {Lax, Peter D},
 doi = {10.1002/cpa.3160100406},
 journal = {{Communications on Pure and Applied Mathematics}},
 number = {4},
 pages = {537--566},
 title = {Hyperbolic systems of conservation laws {{II}}},
 url = {https://doi.org/10.1002/cpa.3160100406},
 volume = {10},
 year = {1957}
}

@article{lee2023convergence,
 archivePrefix = {arXiv},
 author = {Lee, Holden and Lu, Jianfeng and Tan, Yixin},
 doi = {10.48550/arXiv.2209.12381},
 eprint = {2209.12381},
 journal = {{International Conference on Algorithmic Learning Theory (ALT)}},
 title = {Convergence of score-based generative modeling for general data distributions},
 url = {https://arxiv.org/abs/2209.12381},
 year = {2023}
}

@article{li2024critical,
 archivePrefix = {arXiv},
 author = {Li, Mingyang and Chen, Sitan},
 doi = {10.48550/arXiv.2403.01633},
 eprint = {2403.01633},
 journal = {{International Conference on Machine Learning (ICML)}},
 title = {Critical windows: Non-asymptotic theory for feature emergence in diffusion models},
 url = {https://arxiv.org/abs/2403.01633},
 year = {2024}
}

@article{lipman2023flow,
 archivePrefix = {arXiv},
 author = {Lipman, Yaron and Chen, Ricky TQ and Ben-Hamu, Heli and Nickel, Maximilian and Le, Matt},
 doi = {10.48550/arXiv.2210.02747},
 eprint = {2210.02747},
 journal = {{International Conference on Learning Representations (ICLR)}},
 title = {Flow Matching for Generative Modeling},
 url = {https://arxiv.org/abs/2210.02747},
 year = {2023}
}

@article{liu2023flow,
 archivePrefix = {arXiv},
 author = {Liu, Xingchao and Gong, Chengyue and Liu, Qiang},
 doi = {10.48550/arXiv.2209.03003},
 eprint = {2209.03003},
 journal = {{International Conference on Learning Representations (ICLR)}},
 title = {Flow Straight and Fast: Learning to Generate and Transfer Data with Rectified Flow},
 url = {https://arxiv.org/abs/2209.03003},
 year = {2023}
}

@article{montanari2023sampling,
 archivePrefix = {arXiv},
 author = {Montanari, Andrea},
 doi = {10.48550/arXiv.2305.10690},
 eprint = {2305.10690},
 journal = {{arXiv preprint}},
 title = {Sampling, Diffusions, and Stochastic Localization},
 url = {https://arxiv.org/abs/2305.10690},
 year = {2023}
}

@article{rankine1870thermodynamic,
 author = {Rankine, William John Macquorn},
 doi = {10.1098/rstl.1870.0017},
 journal = {{Philosophical Transactions of the Royal Society of London}},
 pages = {277--288},
 title = {On the thermodynamic theory of waves of finite longitudinal disturbance},
 url = {https://doi.org/10.1098/rstl.1870.0017},
 volume = {160},
 year = {1870}
}

@article{raya2023spontaneous,
 archivePrefix = {arXiv},
 author = {Raya, Gabriel and Ambrogioni, Luca},
 doi = {10.48550/arXiv.2305.19693},
 eprint = {2305.19693},
 journal = {{Advances in Neural Information Processing Systems (NeurIPS)}},
 title = {Spontaneous symmetry breaking in generative diffusion models},
 url = {https://arxiv.org/abs/2305.19693},
 volume = {36},
 year = {2023}
}

@article{rombach2022high,
 archivePrefix = {arXiv},
 author = {Rombach, Robin and Blattmann, Andreas and Lorenz, Dominik and Esser, Patrick and Ommer, Bj{\"o}rn},
 doi = {10.1109/CVPR52688.2022.01042},
 eprint = {2112.10752},
 journal = {{IEEE Conference on Computer Vision and Pattern Recognition (CVPR)}},
 pages = {10684--10695},
 title = {High-resolution Image Synthesis with Latent Diffusion Models},
 url = {https://doi.org/10.1109/CVPR52688.2022.01042},
 year = {2022}
}

@article{sclocchi2024phase,
 archivePrefix = {arXiv},
 author = {Sclocchi, Antonio and Favero, Alessandro and Wyart, Matthieu},
 doi = {10.48550/arXiv.2402.16991},
 eprint = {2402.16991},
 journal = {{Proceedings of the National Academy of Sciences}},
 title = {A Phase Transition in Diffusion Models Reveals the Hierarchical Nature of Data},
 url = {https://arxiv.org/abs/2402.16991},
 year = {2024}
}

@article{sohl2015deep,
 author = {Sohl-Dickstein, Jascha and Weiss, Eric and Maheswaranathan, Niru and Ganguli, Surya},
 doi = {10.5555/3045118.3045358},
 journal = {{Journal of Machine Learning Research (Proceedings of ICML 2015)}},
 pages = {2256--2265},
 publisher = {JMLR.org},
 title = {Deep Unsupervised Learning using Nonequilibrium Thermodynamics},
 url = {https://dl.acm.org/doi/10.5555/3045118.3045358},
 year = {2015}
}

@article{song2019generative,
 archivePrefix = {arXiv},
 author = {Song, Yang and Ermon, Stefano},
 doi = {10.48550/arXiv.1907.05600},
 eprint = {1907.05600},
 journal = {{Advances in Neural Information Processing Systems (NeurIPS)}},
 pages = {11895--11907},
 title = {Generative Modeling by Estimating Gradients of the Data Distribution},
 url = {https://arxiv.org/abs/1907.05600},
 volume = {32},
 year = {2019}
}

@article{song2020improved,
 archivePrefix = {arXiv},
 author = {Song, Yang and Ermon, Stefano},
 doi = {10.48550/arXiv.2006.09011},
 eprint = {2006.09011},
 journal = {{Advances in Neural Information Processing Systems (NeurIPS)}},
 title = {Improved Techniques for Training Score-Based Generative Models},
 url = {https://arxiv.org/abs/2006.09011},
 volume = {33},
 year = {2020}
}

@article{song2021ddim,
 archivePrefix = {arXiv},
 author = {Song, Jiaming and Meng, Chenlin and Ermon, Stefano},
 doi = {10.48550/arXiv.2010.02502},
 eprint = {2010.02502},
 journal = {{International Conference on Learning Representations (ICLR)}},
 title = {Denoising Diffusion Implicit Models},
 url = {https://arxiv.org/abs/2010.02502},
 year = {2021}
}

@article{song2021score,
 archivePrefix = {arXiv},
 author = {Song, Yang and Sohl-Dickstein, Jascha and Kingma, Diederik P and Kumar, Abhishek and Ermon, Stefano and Poole, Ben},
 doi = {10.48550/arXiv.2011.13456},
 eprint = {2011.13456},
 journal = {{International Conference on Learning Representations (ICLR)}},
 title = {Score-based Generative Modeling through Stochastic Differential Equations},
 url = {https://arxiv.org/abs/2011.13456},
 year = {2021}
}

@article{tang2024score,
 archivePrefix = {arXiv},
 author = {Tang, Wenpin and Zhao, Hanyang},
 doi = {10.48550/arXiv.2402.07487},
 eprint = {2402.07487},
 journal = {{arXiv preprint}},
 title = {Score-based diffusion models via stochastic differential equations -- a technical tutorial},
 url = {https://arxiv.org/abs/2402.07487},
 year = {2024}
}

@book{tsybakov2009introduction,
 author = {Tsybakov, Alexandre B},
 publisher = {Springer},
 title = {Introduction to Nonparametric Estimation},
 year = {2009}
}

@article{vincent2011connection,
 author = {Vincent, Pascal},
 doi = {10.1162/NECO_a_00142},
 journal = {{Neural Computation}},
 pages = {1661--1674},
 title = {A connection between score matching and denoising autoencoders},
 url = {https://doi.org/10.1162/NECO_a_00142},
 volume = {23},
 year = {2011}
}

@article{vuong2025score,
 archivePrefix = {arXiv},
 author = {Vuong, An B and Lin, Yen Ting and others},
 doi = {10.48550/arXiv.2509.00336},
 eprint = {2509.00336},
 journal = {{Transactions on Machine Learning Research}},
 title = {Are We Really Learning the Score Function? {{R}}einterpreting Diffusion Models Through {{W}}asserstein Gradient Flow Matching},
 url = {https://arxiv.org/abs/2509.00336},
 year = {2025}
}

@book{whitham1974linear,
 author = {Whitham, Gerald Beresford},
 doi = {10.1002/9781118032954},
 publisher = {John Wiley \& Sons},
 title = {Linear and Nonlinear Waves},
 url = {https://doi.org/10.1002/9781118032954},
 year = {1974}
}

\end{document}